\documentclass[letterpaper,10 pt,conference]{ieeeconf}
\IEEEoverridecommandlockouts 
\usepackage{amsmath,amssymb,amsfonts}
\usepackage{graphicx}
\usepackage{subfig}
\usepackage{textcomp}
\usepackage{color}
\usepackage{xcolor}
\usepackage{multirow}
\usepackage{multicol,lipsum}
\usepackage{epsfig}
\usepackage{algorithm}
\usepackage{algorithmic}
\usepackage{multirow}
\usepackage{hyperref}
\usepackage{mathtools}
\usepackage{bm}
\usepackage{listings}
\newtheorem{remark}{Remark}
\newtheorem{definition}{Definition}
\newtheorem{theorem}{Theorem}

\newtheorem{lemma}{Lemma}

\usepackage{array}
\usepackage{float}
\usepackage[short]{optidef}
\usepackage{mdframed}
\mdfdefinestyle{mystyle}{
    backgroundcolor=white!10}

\newcommand{\oomit}[1]{}

\newcolumntype{M}[1]{>{\centering\arraybackslash}m{#1}}
\newcolumntype{N}{@{}m{0pt}@{}}

\begin{document}

\title{Safe Exit Controllers Synthesis for Continuous-time Stochastic Systems}

\author{Bai Xue
\thanks{State Key Lab. of Computer Science, Institute of Software, CAS, Beijing, China \{xuebai@ios.ac.cn\} 
}
\thanks{University of Chinese Academy of Sciences, Beijing, China}
}

\maketitle
\thispagestyle{empty}
\pagestyle{empty}

\begin{abstract}
This paper tackles the problem of generating safe exit controllers for continuous-time systems described by stochastic differential equations (SDEs). The primary aim is to develop controllers that maximize the lower bounds of the exit probability that the system escapes from a safe but uncomfortable set within a specified time frame and guide it towards a comfortable set. The paper considers two distinct cases: one in which the boundary of the safe set is a subset of the boundary of the uncomfortable set, and the other where the boundaries of the two sets do not intersect. To begin, we present a sufficient condition for establishing lower bounds on the exit probability in the first case. This condition serves as a guideline for constructing an online linear programming problem. The linear programming problem is designed to implicitly synthesize an optimal exit controller that maximizes the lower bounds of the exit probability. The method employed in the first case is then extended to the second one. Finally, we demonstrate the effectiveness of the proposed approaches on one example.
\end{abstract}

\section{Introduction}
\label{sec:intro}
Stochastic systems are highly significant in various fields such as robotics, finance, and biology due to their ability to model uncertain factors that can greatly influence system behavior. Stochastic differential equations (SDEs) provide a powerful modeling approach for such systems as they allow for the incorporation of inherent uncertainties in system dynamics \cite{kloeden1992stochastic}. This enables the analysis of system behavior, as well as the verification of properties related to safety, reliability, and performance.

In recent years, there has witnessed an increased focus on safety properties \cite{franzle2011measurability,tomlin2003computational,abate2010approximate}, particularly in the context of safety-critical systems. Safety verification via barrier certificates for stochastic systems with infinite time horizons was introduced in \cite{prajna2007framework} alongside the deterministic counterpart. This framework builds upon the known Doob's nonnegative supermartingale inequality (or, Ville's inequality \cite{ville1939etude}) and enables bounding the exit probability from above, indicating the likelihood of a system leaving a safe region. However, this approach has a limitation as it requires the infinitesimal generator, responsible for the expected value evolution of a stochastic process, to be non-positive. Consequently, the barrier function is restricted to be a supermartingale. To overcome this restriction, \cite{steinhardt2012finite} relaxed the condition by introducing barrier certificates based on c-martingales. A c-martingale allows the expected value of the barrier function to increase over time while providing an upper bound on the infinitesimal generator. This approach provides upper bounds of the exit probability for systems with finite time horizons. Afterwards, inspired by studies in \cite{kushner1967stochastic}, \cite{santoyo2021barrier} enhanced the c-martingales and proposed a barrier certificate constraint that imposes a state-dependent bound on the infinitesimal generator for upper-bounding the exit probability with finite time horizons. Moreover, a  sum-of-squares optimization based method was proposed in \cite{santoyo2021barrier}
to synthesizing polynomial state feedback controllers. Further contributions to the computation of upper bounds of the exit probability include \cite{nilsson2020lyapunov}, which presented a comparison theorem for one-dimensional SDEs and applied it to upper-bound exit probabilities for multi-dimensional SDEs in terms of an exit probability of a one-dimensional process. Recently, based on online convex quadratic programs that synthesize controllers implicitly \cite{freeman1996inverse,ames2016control}, \cite{wang2021safety} introduced stochastic control barrier functions as a framework for synthesizing controllers that enforce upper bounds on exit probabilities over both infinite and finite time horizons. The conditions for upper-bounding exit probabilities in the aforementioned works, except \cite{nilsson2020lyapunov}, are constructed or derived from the  Doob's nonnegative supermartingale inequality. 

On the other hand, in \cite{xue2021reach}, a novel approach was proposed for characterizing the exact reachability probability of discrete-time stochastic systems. This probability measures the likelihood of a system starting from an initial set and eventually entering target sets, while staying within safe sets before the first target hitting time. Unlike previous methods that rely on Doob's nonnegative supermartingale inequality, this approach derives an equation that provides an exact estimation of the reachability probability \cite{xue2023reachability}. By relaxing this equation, barrier-like conditions can be obtained to both lower-bound and upper-bound the reachability probability. Additionally, the method has been extended in \cite{xue2023cdc} to compute lower and upper bounds of the exit probability over an infinite time horizon for discrete-time stochastic systems. Furthermore, the equation and its relaxations have been further extended in \cite{xue2022reach} to perform reach-avoid analysis over infinite-time horizons for systems modeled by SDEs. The use of sum-of-squares optimization techniques has enabled the application of these barrier-like conditions in the synthesis of controllers for safety-critical systems, as in \cite{xue2023reach}.

 In safety-focused applications, it is common to prioritize the computation of upper bounds for the exit probability from a safe set. However, there is a significant lack of methods specifically focused on computing lower bounds, despite their significance in certain practical scenarios. Consider a situation where a system operates within a safe set but experiences discomfort, such as a robotic system navigating around the boundary of the safe set.  Although the system is safe, it may encounter discomfort due to the fragility of safety violations. In this situation, the system would prefer to leave this typical safe set to alleviate the discomfort. By maximizing lower bounds of the exit probability, we can ensure that the system has a higher probability of safely leaving this uncomfortable set and reaching a safe set that provides more comfort. It not only ensures safety but also considers comfort, resulting in a more holistic solution for safety-focused applications. This aspect becomes increasingly important for systems like autonomous vehicles, where comfort plays a substantial role once safety requirements are met. Additionally, considering lower bounds can complement existing methods that focus on computing upper bounds of the exit probability, and thus can provide us a more comprehensive analysis of the system's behavior.

In this paper, we investigate the problem of generating safe controllers that optimize the lower bounds of exit probabilities for continuous-time systems represented by SDEs. The exit probability refers to the likelihood of a system, starting from an open, safe but uncomfortable set (which is a subset of the safe set), exiting that set within a specified time frame and entering a comfortable set. This time frame can either be finite or infinite. We analyze two different cases in this study. In the first case, the boundary of the safe set is a subset of the boundary of the uncomfortable set. We begin by establishing a sufficient condition for lower-bounding the exit probability in this case, extending the condition presented in \cite{xue2023reach}. Based on the proposed sufficient condition, we formulate an online linear programming problem to synthesize an optimal controller implicitly that maximizes lower bounds of the exit probability. Then we extend the sufficient condition and linear programming method in the first case to the second one, in which the boundary of the safe set does not intersect with the boundary of the uncomfortable set. Finally, to illustrate the effectiveness of our proposed methods, we provide an example application and demonstrate their applicability.

The main contribution of our work is summarized as follows: unlike previous studies that primarily focused on synthesizing controllers to enforce upper bounds on the exit probability for systems modeled by SDEs, the present work introduces novel conditions for controller synthesis that specifically provide lower bounds of the exit probability. These conditions are applicable to both finite and infinite time horizons in exit analysis. One key aspect of our contribution is that our proposed conditions not only extend the existing condition presented in \cite{xue2023reach} to the finite-time scenario but also encompass it as a special case within our framework. This demonstrates the versatility of our conditions in handling a wider range of scenarios compared to the one in \cite{xue2023reach}.

This paper is structured as follows. In Section \ref{sec:pre}, we introduce SDEs and the problems of synthesizing safe exit controllers. In Section \ref{sec:ecs}, we present our sufficient conditions for characterizing lower bounds of the exit probabilities
and our linear programming methods for synthesizing controllers that maximize these lower bounds. In Section \ref{sec:expe}, we demonstrate the effectiveness of our approach through one example. Finally, in Section \ref{sec:con}, we conclude the paper and discuss avenues for future research.

Some basic notions are used in this paper: $\mathbb{R}$ and $\mathbb{R}_{\geq 0}$ stand for the set of real numbers and non-negative real numbers, respectively; $\mathbb{R}^{n}$ and $\mathbb{R}^{n\times m}$ denote the space of all $n$-dimensional vectors and $n\times m$ real matrices, respectively; for a set $\mathcal{A}$, $\overline{\mathcal{A}}$ and $\partial \mathcal{A}$ denotes the closure and boundary of the set $\mathcal{A}$, respectively; $\wedge$ denotes the logical operation of conjunction.

\section{Preliminaries}
\label{sec:pre}

This section introduces SDEs and the exit controllers synthesis problem of interest.

Consider an affine stochastic control system, 
\begin{equation}
    \label{SDE}
    \begin{split}
d\bm{x}(t,\bm{w})=(\bm{f}_1(\bm{x}(t,\bm{w}))+&\bm{f}_2(\bm{x}(t,\bm{w}))\bm{u}(\bm{x}(t)))dt\\
&+\bm{\sigma}(\bm{x}(t,\bm{w})) d\bm{W}(t,\bm{w}),
\end{split}
\end{equation}
where $\bm{f}_1(\cdot):\mathbb{R}^n\rightarrow \mathbb{R}^n$, $\bm{f}_2(\cdot): \mathbb{R}^n \rightarrow \mathbb{R}^{n\times m}$, and $\bm{\sigma}(\cdot): \mathbb{R}^n \rightarrow \mathbb{R}^{n\times k}$ are locally Lipschitz continuous function; the admissible input is defined by the function $\bm{u}(\cdot): \mathbb{R}^n \rightarrow \mathcal{U}$ with $\mathcal{U}$ being the admissible input set; $\bm{W}(t,\bm{w}): \mathbb{R}\times \Omega\rightarrow \mathbb{R}^k$ is an $k$-dimensional Wiener process (standard Brownian motion), and $\Omega$, equipped with the probability measure $\mathbb{P}$,  is the sample space $\bm{w}$ belongs to.  The expectation with respect to $\mathbb{P}$ is denoted by $\mathbb{E}[\cdot]$.

Given a locally Lipschitz controller $\bm{u}(\bm{x})$, for an initial state $\bm{x}_0$, the SDE \eqref{SDE} has a unique (maximal local) strong solution over a time interval $[0,T^{\bm{x}_0}(\bm{w}))$, where $T^{\bm{x}_0}(\bm{w})$ is a positive real value or infinity. This solution is denoted as $\bm{\phi}_{\bm{x}_0}^{\bm{w}}(\cdot): [0,T^{\bm{x}_0}(\bm{w}))\rightarrow \mathbb{R}^n$, which satisfies the stochastic integral equation,
\begin{equation*}
\begin{split}
  \bm{\phi}_{\bm{x}_0}^{\bm{w}}(t)&=\int_{0}^t (\bm{f}_1(\bm{\phi}_{\bm{x}_0}^{\bm{w}}(\tau))+\bm{f}_2(\bm{\phi}_{\bm{x}_0}^{\bm{w}}(\tau))\bm{u}(\bm{\phi}_{\bm{x}_0}^{\bm{w}}(\tau)))d \tau\\
   &+\int_{0}^t \bm{\sigma}(\bm{\phi}_{\bm{x}_0}^{\bm{w}}(\tau)) d\bm{W}(\tau,\bm{w})+\bm{x}_0.
   \end{split}
\end{equation*}

The infinitesimal generator underlying system \eqref{SDE}, which represents the limit of the expected value of $v(\bm{\phi}_{\bm{x}_0}^{\bm{w}}(t))$ as $t$ approaches 0, is presented in Definition \ref{inf_generator}.
\begin{definition}
\label{inf_generator}
Given system \eqref{SDE}  with a locally Lipschitz controller $\bm{u}(\bm{x})$,  the infinitesimal generator of a twice continuously differentiable function  $v(\bm{x})$ is defined by 
\begin{equation*}
    \begin{split}
        &\mathcal{L}_{v,\bm{u}}(\bm{x}_0)=\lim_{t\rightarrow 0}\frac{\mathbb{E}[v(\bm{\phi}_{\bm{x}_0}^{\bm{w}}(t))]-v(\bm{x}_0)}{t}=\\
        &[\frac{\partial v}{\partial \bm{x}}(\bm{f}_1(\bm{x})+\bm{f}_2(\bm{x})\bm{u}(\bm{x}))+\frac{1}{2}\textbf{tr}(\bm{\sigma}(\bm{x})^{\top}\frac{\partial^2 v}{\partial \bm{x}^2} \bm{\sigma}(\bm{x}))]\mid_{\bm{x}=\bm{x}_0},
    \end{split}
\end{equation*}
where $\frac{\partial v}{\partial \bm{x}}$ represents the gradient of the function $v(\bm{x})$ with respect to $\bm{x}$, and $\textbf{tr}(\cdot)$ denotes the trace of a matrix.
\end{definition}

Given a safe set $\mathcal{S}\subseteq \mathbb{R}^n$ and an  uncomfortable set $\mathcal{C} \subseteq \mathcal{S}$, a safe exit controller is a controller that maximizes the exit probability of system \eqref{SDE}, starting from $\mathcal{C}$, entering the comfortable set $\mathcal{S}\setminus \mathcal{C}$ within a specified time horizon. Additionally, it is required that the system remains inside $\mathcal{C}$ before leaving it.

\begin{definition}[Safe Exit Controllers]
\label{RAC_s}
Given a time horizon $\mathbb{T}$, an initial state $\bm{x}_0\in \mathcal{C}$ and a probability threshold $p_{\bm{x}_0}\in[0,1]$, an exit controller is a locally Lipschitz controller $\bm{u}(\cdot): \overline{\mathcal{C}}\rightarrow \mathbb{R}^m$ satisfying the following condition:
\begin{equation}
\mathbb{P}\Bigg(\Big\{\bm{w}\in \Omega \mid
\begin{aligned}
&\exists t\in \mathbb{T}. \bm{\phi}_{\bm{x}_0}^{\bm{w}}(t) \in \mathcal{S}\setminus \mathcal{C}\bigwedge\\
&\forall \tau \in [0,t). \bm{\phi}_{\bm{x}_0}^{\bm{w}}(\tau) \in \mathcal{C}
\end{aligned}
\Big\}
\Bigg)\geq p_{\bm{x}_0},
\end{equation}
where $\mathbb{T}=[0,T]$ if $T<\infty$, and $\mathbb{T}=[0,\infty)$ otherwise.
\end{definition}

In Definition \ref{RAC_s}, the exit controller is related to a lower bound of the exact exit probability. The safe exit controllers synthesis problem of interest in this work is to synthesize an exit controller maximizing the threshold $p_{\bm{x}_0}$. The safe exit controller synthesis problem in this paper is considered in the following two distinct cases.

The first case we consider is that the boundary of the safe set $\mathcal{S}$ is a subset of the one of the uncomfortable set $\mathcal{C}$, i.e., $\partial \mathcal{S} \subseteq \partial \mathcal{C}$. Specifically, we assume $\mathcal{S}=\{\bm{x}\in \mathbb{R}^n\mid h(\bm{x})>0\}$ with $\partial \mathcal{S}=\{\bm{x}\in \mathbb{R}^n\mid h(\bm{x})=0\}$ and $\mathcal{C}=\{\bm{x}\in \mathbb{R}^n\mid 0<h(\bm{x})<1\}$ with $\partial \mathcal{C}=\{\bm{x}\in \mathbb{R}^n\mid h(\bm{x})=0\vee h(\bm{x})=1\}$. 
This assumption is made based on the practical consideration that a system operating close to the boundary of a safe set is at a higher risk of safety hazards, thereby making the system operation in this set uncomfortable. In this case, system \eqref{SDE} should be enforced to exit the set $\mathcal{C}$ through states satisfying $h(\bm{x})=1$ rather than $h(\bm{x})=0$. Thus, that $\exists t\in \mathbb{T}. \bm{\phi}_{\bm{x}_0}^{\bm{w}}(t) \in \mathcal{S}\setminus \mathcal{C} \wedge  \forall \tau \in [0,t). \bm{\phi}_{\bm{x}_0}^{\bm{w}}(\tau) \in \mathcal{C}$ is equivalent to 
$\exists t\in \mathbb{T}. h(\bm{\phi}_{\bm{x}_0}^{\bm{w}}(t))=1 \wedge  \forall \tau \in [0,t). \bm{\phi}_{\bm{x}_0}^{\bm{w}}(\tau) \in \mathcal{C}$. The corresponding exit controllers synthesis problem  is formulated in Definition \ref{RAC_p}.

\begin{definition}[Safe Exit Controllers Synthesis Problem I]
\label{RAC_p}
 Assume the safe set is $\mathcal{S}=\{\bm{x}\in \mathbb{R}^n\mid h(\bm{x})>0\}$ with $\partial \mathcal{S}=\{\bm{x}\in \mathbb{R}^n\mid h(\bm{x})=0\}$ and the uncomfortable set $\mathcal{C}=\{\bm{x}\in \mathbb{R}^n\mid 0<h(\bm{x})<1\}$ with $\partial \mathcal{C}=\{\bm{x}\in \mathbb{R}^n\mid h(\bm{x})=0\vee h(\bm{x})=1\}$, where $h(\cdot): \mathbb{R}^n \rightarrow \mathbb{R}$ is a twice continuously differentiable function. Given a time horizon $\mathbb{T}$, the safe exit controllers synthesis problem is to synthesize a locally Lipschitz controller $\bm{u}(\cdot): \overline{\mathcal{C}}\rightarrow \mathbb{R}^m$ of maximizing lower bounds of the exit probability for system \eqref{SDE} leaving the set $\mathcal{C}$ through states in $\{\bm{x}\in \mathbb{R}^n\mid h(\bm{x})=1\}$, i.e.,  solving the following optimization problem:
\begin{equation}
  \begin{split}
    &\max_{\bm{u}}  p_{\bm{x}_0}\\
    &\text{s.t.~} \mathbb{P}\Bigg(\Big\{\bm{w}\in \Omega \mid
\begin{split}
&\exists t\in \mathbb{T}. h(\bm{\phi}_{\bm{x}_0}^{\bm{w}}(t))=1\bigwedge\\
&\forall \tau \in [0,t). \bm{\phi}_{\bm{x}_0}^{\bm{w}}(\tau) \in \mathcal{C}
\end{split}
\Big\}
\Bigg)\geq p_{\bm{x}_0},
    \end{split}
\end{equation}
where $\mathbb{T}=[0,T]$ if $T<\infty$, and $\mathbb{T}=[0,\infty)$ otherwise.
\end{definition}

The second case we consider is that the boundary of the uncomfortable set $\mathcal{C}$ does not intersect the boundary of the safe set $\mathcal{S}$, i.e., $\partial \mathcal{S} \cap \partial \mathcal{C}=\emptyset$. In this case, we assume $\mathcal{S}=\{\bm{x}\in \mathbb{R}^n\mid h(\bm{x})>0\}$ with $\partial \mathcal{S}=\{\bm{x}\in \mathbb{R}^n\mid h(\bm{x})=0\}$ and $\mathcal{C}=\{\bm{x}\in \mathbb{R}^n\mid g(\bm{x})<1\}$ with $\partial \mathcal{C}=\{\bm{x}\in \mathbb{R}^n\mid g(\bm{x})=1\}$. 
In this case, that $\exists t\in [0,T]. \bm{\phi}_{\bm{x}_0}^{\bm{w}}(t) \in \mathcal{S}\setminus \mathcal{C} \wedge  \forall \tau \in [0,t). \bm{\phi}_{\bm{x}_0}^{\bm{w}}(\tau) \in \mathcal{C}$ is equivalent to  
$\exists t\in [0,T]. \bm{\phi}_{\bm{x}_0}^{\bm{w}}(t)\in \partial \mathcal{C} \wedge  \forall \tau \in [0,t). \bm{\phi}_{\bm{x}_0}^{\bm{w}}(\tau) \in \mathcal{C}$. Thus, the corresponding safe exit controllers synthesis problem is formulated in Definition \ref{RAC_p1}.

\begin{definition}[Safe Exit Controllers Synthesis Problem II]
\label{RAC_p1}
Assume the uncomfortable set is $\mathcal{C}=\{\bm{x}\in \mathbb{R}^n\mid g(\bm{x})<1\}$ with $\partial \mathcal{C}=\{\bm{x}\in \mathbb{R}^n\mid g(\bm{x})=1\}$ and $\partial \mathcal{S}\cap \partial \mathcal{C}=\emptyset$, where $g(\cdot): \mathbb{R}^n \rightarrow \mathbb{R}$ is a twice continuously differentiable function. Given a time horizon $\mathbb{T}$, the safe exit controllers synthesis problem is to synthesize a locally Lipschitz controller $\bm{u}(\cdot): \overline{\mathcal{C}}\rightarrow \mathbb{R}^m$ of maximizing lower bounds of the exit probability, i.e.,  solving the following optimization problem:
\begin{equation}
  \begin{split}
    &\max_{\bm{u}}  p_{\bm{x}_0}\\
    &\text{s.t.~} \mathbb{P}\Bigg(\Big\{\bm{w}\in \Omega \mid
\begin{split}
&\exists t\in \mathbb{T}. g(\bm{\phi}_{\bm{x}_0}^{\bm{w}}(t))=1 \wedge \\
&\forall \tau \in [0,t). \bm{\phi}_{\bm{x}_0}^{\bm{w}}(\tau) \in \mathcal{C}
\end{split}
\Big\}
\Bigg)\geq p_{\bm{x}_0},
    \end{split}
\end{equation}
where $\mathbb{T}=[0,T]$ if $T<\infty$, and $\mathbb{T}=[0,\infty)$ otherwise.
\end{definition}

\section{Exit Controllers Synthesis}
\label{sec:ecs}
In this section, we describe our approach to solving the safe exit controllers synthesis problems I and II. We first focus on Problem I in Subsection \ref{sec:ecsc}, where we present a condition that exit controllers satisfy in order to derive lower bounds on the exit probabilities for both infinite and finite time horizons. This condition involves two free parameters that need to be optimized. Then, we extend this condition to Problem II in Subsection \ref{sec:ecec1}. Finally, in Subsection \ref{sec:lpbc}, we construct linear programs to synthesize optimal exit controllers implicitly. By optimizing the two free parameters from the conditions, we can design exit controllers that maximize the lower bounds on the exit probabilities. These linear programs enable us to perform online synthesis of the optimal exit controllers for both Problems I and II.

\subsection{Safe Exit Controllers Synthesis Conditions for Problem I}
\label{sec:ecsc}

This subsection introduces a condition that exit controllers satisfy in order to derive lower bounds on the exit probabilities in Problem I for both infinite and finite time horizons.

The construction of the condition lies on an auxiliary stochastic process $\{\widetilde{\bm{\phi}}_{\bm{x}_0}^{\bm{w}}(t), t\in \mathbb{R}_{\geq 0}\}$ for $\bm{x}_0\in \overline{\mathcal{C}}$ that is a stopped process corresponding to $\{\bm{\phi}_{\bm{x}_0}^{\bm{w}}(t), t\in [0,T^{\bm{x}_0}(\bm{w}))\}$ and the set $\mathcal{C}$, i.e., 
\begin{equation}
\widetilde{\bm{\phi}}_{\bm{x}_0}^{\bm{w}}(t)=
\begin{cases}
&\bm{\phi}_{\bm{x}_0}^{\bm{w}}(t), \text{\rm~if~}t<\tau^{\bm{x}_0}(\bm{w}),\\
&\bm{\phi}_{\bm{x}_0}^{\bm{w}}(\tau^{\bm{x}_0}(\bm{w})), \text{\rm~if~}t\geq \tau^{\bm{x}_0}(\bm{w}),
\end{cases}
\end{equation}
where \[\tau^{\bm{x}_0}(\bm{w})=\inf\{t\mid \bm{\phi}_{\bm{x}_0}^{\bm{w}}(t) \in \partial \mathcal{C}\}\] is the first time of exit of $\bm{\phi}_{\bm{x}_0}^{\bm{w}}(t)$ from the open set $\mathcal{C}$. It is worth remarking here that if the path $\bm{\phi}_{\bm{x}_0}^{\bm{w}}(t)$ escapes to infinity in finite time, it must touch the boundary of the set $\mathcal{C}$ and thus $\tau^{\bm{x}_0}(\bm{w})\leq T^{\bm{x}_0}(\bm{w})$. The stopped process $\widehat{\bm{\phi}}_{\bm{x}_0}^{\bm{w}}(t)$ inherits the right continuity and strong Markovian property of $\bm{\phi}_{\bm{x}_0}^{\bm{w}}(t)$. Moreover, the infinitesimal generator corresponding to $\widehat{\bm{\phi}}_{\bm{x}_0}^{\bm{w}}(t)$ is identical to the one corresponding to $\bm{\phi}_{\bm{x}_0}^{\bm{w}}(t)$ over $\mathcal{X}$, and is equal to zero on the boundary $\partial \mathcal{C}$ \cite{kushner1967}. That is, for $v(\bm{x})$ being a twice continuously differentiable function, 
\[
\begin{split}
\widetilde{\mathcal{L}}_{v,\bm{u}}(\bm{x})=\mathcal{L}_{v,\bm{u}}(\bm{x})&=\frac{\partial v}{\partial \bm{x}}(\bm{f}_1(\bm{x})+\bm{f}_2(\bm{x})\bm{u}(\bm{x}))\\
&+\frac{1}{2}\textbf{tr}(\bm{\sigma}(\bm{x})^{\top}\frac{\partial^2 v}{\partial \bm{x}^2} \bm{\sigma}(\bm{x}))
\end{split}
\] for $\bm{x}\in \mathcal{C}$ and $\widetilde{\mathcal{L}}_{v,\bm{u}}(\bm{x})=0$ for $\bm{x}\in \partial \mathcal{C}$.

The probability of reaching the set $\mathcal{C}_1$ within the time horizon $\mathbb{T}=[0,T]$ for system \eqref{SDE} while staying inside the set $\mathcal{C}$ before the first time of hitting $\mathcal{C}_1$, is equal to the probability of reaching the set $\mathcal{C}_1$ at the time instant $T$ for the auxiliary stochastic process, where $\mathcal{C}_1=\{\bm{x}\in \mathbb{R}^n \mid h(\bm{x})=1\}$.
\begin{lemma}
\label{equiv}
    Given a time instant $T>0$ and $\bm{x}_0\in \mathcal{C}$, \[
    \begin{split}
   & \mathbb{P}(\exists t\in [0,T]. \bm{\phi}_{\bm{x}_0}^{\bm{w}}(t)\in \mathcal{C}_1 \wedge 
 \forall \tau\in [0,t). \bm{\phi}_{\bm{x}_0}^{\bm{w}}(\tau)\in \mathcal{C})\\
 &=\mathbb{P}(\widetilde{\bm{\phi}}_{\bm{x}_0}^{\bm{w}}(T)\in \mathcal{C}_1)=\mathbb{E}[1_{\mathcal{C}_1}(\widetilde{\bm{\phi}}_{\bm{x}_0}^{\bm{w}}(T))].
 \end{split}
 \] Moreover, for any $0<T_1\leq T_2$, \[\mathbb{P}(\widetilde{\bm{\phi}}_{\bm{x}_0}^{\bm{w}}(T_1)\in \mathcal{C}_1)\leq \mathbb{P}(\widetilde{\bm{\phi}}_{\bm{x}_0}^{\bm{w}}(T_2)\in \mathcal{C}_1),\]
where $\mathcal{C}_1=\{\bm{x}\in \mathbb{R}^n \mid h(\bm{x})=1\}$. 
\end{lemma}
\begin{proof}
 It is easy to observe that $\{\bm{w}\in \Omega\mid \exists t\in [0,T]. \bm{\phi}_{\bm{x}_0}^{\bm{w}}(t)\in \mathcal{C}_1 \wedge 
 \forall \tau\in [0,t). \bm{\phi}_{\bm{x}_0}^{\bm{w}}(\tau)\in \mathcal{C}\}=\{\bm{w}\in \Omega\mid \widetilde{\bm{\phi}}_{\bm{x}_0}^{\bm{w}}(T)\in \mathcal{C}_1\}$.  Therefore, the conclusion holds. 

 In addition, we observe that for $T_1\leq T_2$,  
 \[\{\bm{w}\in \Omega\mid \widetilde{\bm{\phi}}_{\bm{x}_0}^{\bm{w}}(T_1)\in \mathcal{C}_1\}\subseteq \{\bm{w}\in \Omega\mid \widetilde{\bm{\phi}}_{\bm{x}_0}^{\bm{w}}(T_2)\in \mathcal{C}_1\}.\] Consequently, \[\mathbb{P}(\widetilde{\bm{\phi}}_{\bm{x}_0}^{\bm{w}}(T_1)\in \mathcal{C}_1)\leq \mathbb{P}(\widetilde{\bm{\phi}}_{\bm{x}_0}^{\bm{w}}(T_2)\in \mathcal{C}_1)\] holds for $T_1\leq T_2$.
\end{proof}

\begin{remark}
The conclusion that 
\begin{equation*}
\begin{split}
&\mathbb{P}(\exists t\in [0,\infty). \bm{\phi}_{\bm{x}_0}^{\bm{w}}(t)\in \mathcal{C}_1 \wedge \forall \tau \in [0,t). \bm{\phi}_{\bm{x}_0}^{\bm{w}}(\tau)\in \mathcal{C})\\
&=\lim_{t\rightarrow \infty}\mathbb{P}(\widetilde{\bm{\phi}}_{\bm{x}_0}^{\bm{w}}(t)\in \mathcal{C}_1)
\end{split}
\end{equation*}
is shown in \cite{xue2022reach}, where $\bm{x}_0\in \mathcal{C}$.
\end{remark}

Based on the auxiliary stochastic process defined above, a condition can be straightforwardly obtained from Proposition 3 in \cite{xue2023reach} to lower-bound the exit probability over the infinite time horizon. 

\begin{lemma}
\label{previous}
If there exists a locally Lipschitz controller $\bm{u}(\cdot): \overline{\mathcal{C}}\rightarrow \mathcal{U}$ satisfying the following condition:
\begin{equation}
\label{stochastic_c_e00}
\mathcal{L}_{h,\bm{u}}(\bm{x})\geq  a h(\bm{x}),  \forall \bm{x}\in \mathcal{C},
\end{equation}
where $a>0$, then \[\mathbb{P}(\exists t\geq 0. h(\bm{\phi}_{\bm{x}_0}^{\bm{w}}(t))=1 \wedge \forall \tau \in [0,t). \bm{\phi}_{\bm{x}_0}^{\bm{w}}(\tau) \in \mathcal{C})\geq h(\bm{x}_0).\]
\end{lemma}

Lemma \ref{previous} introduces a useful condition that includes a free parameter $a$. This condition is designed to establish a lower bound on the exit probability for Problem I over an infinite time horizon. However, the lower bound provided by Lemma \ref{previous} is solely determined by the initial state of the system described in Equation \eqref{SDE}, and it does not rely on the value of $a$. Therefore, optimizing the value of $a$ does not impact the lower bound on the exit probability for Problem I over the infinite time horizon. Moreover, condition \eqref{stochastic_c_e00} may be overly stringent, significantly constraining the feasible space for the controller $\bm{u}(\cdot): \overline{\mathcal{C}}\rightarrow \mathcal{U}$. Below, we will introduce an additional parameter $b$ into condition \eqref{stochastic_c_e00} to establish a more general and less restrictive condition that can provide lower bounds on exit probabilities for both finite and infinite time horizons in Problem I.

\begin{theorem}
\label{exponential_s2}
If there exists a locally Lipschitz controller $\bm{u}(\cdot): \overline{\mathcal{C}}\rightarrow \mathcal{U}$ satisfying the following condition:
\begin{equation}
\label{stochastic_c_e2}
\begin{cases}
&\mathcal{L}_{h,\bm{u}}(\bm{x})\geq  a h(\bm{x})-b,  \forall \bm{x}\in \mathcal{C},\\
&a>b\geq 0,
\end{cases}
\end{equation}
then  
    \[\mathbb{P}_T\geq \max\{0,\frac{e^{aT}(h(\bm{x}_0)-\frac{b}{a})+\frac{b}{a}-1}{(1-\frac{b}{a})(e^{aT}-1)}\}\]
         and \[\mathbb{P}_{\infty}\geq \max\{0,\frac{h(\bm{x}_0)-\frac{b}{a}}{1-\frac{b}{a}}\},\]
where $\mathbb{P}_T=\mathbb{P}(\exists t\in [0,T]. h(\bm{\phi}_{\bm{x}_0}^{\bm{w}}(t))=1\wedge \forall \tau \in [0,t). \bm{\phi}_{\bm{x}_0}^{\bm{w}}(\tau) \in \mathcal{C})$ and 
$\mathbb{P}_{\infty}=\mathbb{P}(\exists t\geq 0. h(\bm{\phi}_{\bm{x}_0}^{\bm{w}}(t))=1\wedge \forall \tau \in [0,t). \bm{\phi}_{\bm{x}_0}^{\bm{w}}(\tau) \in \mathcal{C})$.
\end{theorem}
\begin{proof}
    According to \eqref{stochastic_c_e2}, we have 
    \begin{equation}
\begin{cases}
&\widetilde{\mathcal{L}}_{h,\bm{u}}(\bm{x})+(a-b)1_{\mathcal{C}_1}(\bm{x})\geq  a h(\bm{x})-b,  \forall \bm{x}\in \overline{\mathcal{C}},\\
&a>b\geq 0,
\end{cases}
\end{equation} 
where $\mathcal{C}_1=\{\bm{x}\in \mathbb{R}^n \mid h(\bm{x})=1\}$ and \[\widetilde{\mathcal{L}}_{h,\bm{u}}(\bm{x})=\begin{cases}
    &\mathcal{L}_{v,\bm{u}}(\bm{x}), \text{~if~} \bm{x}\in  \mathcal{C},\\
    &0, ~~~~~~~~~\text{~if~} \bm{x}\in \partial \mathcal{C}.
\end{cases}
\]

Consequently, 
\[
\begin{split}
\mathbb{E}&[h(\widetilde{\bm{\phi}}_{\bm{x}_0}^{\bm{w}}(T))]\geq \int_{0}^T a \mathbb{E}[h(\widetilde{\bm{\phi}}_{\bm{x}_0}^{\bm{w}}(t))]dt+h(\bm{x}_0)\\
&~~-\int_{0}^T b dt-\int_{0}^T (a-b)\mathbb{E}[1_{\mathcal{C}_1}(\widetilde{\bm{\phi}}_{\bm{x}_0}^{\bm{w}}(t))]dt, \forall \bm{x}_0\in \mathcal{C}.
\end{split}
\]
Taking $\overline{h}(\bm{x})=-h(\bm{x})$ over $\bm{x}\in \overline{\mathcal{C}}$, we have 
\[
\begin{split}
\mathbb{E}&[\overline{h}(\widetilde{\bm{\phi}}_{\bm{x}_0}^{\bm{w}}(T))]\leq \int_{0}^T a \mathbb{E}[\overline{h}(\widetilde{\bm{\phi}}_{\bm{x}_0}^{\bm{w}}(t))]dt+\overline{h}(\bm{x}_0)\\
&~~+\int_{0}^T b dt+\int_{0}^T (a-b)\mathbb{E}[1_{\mathcal{C}_1}(\widetilde{\bm{\phi}}_{\bm{x}_0}^{\bm{w}}(t))]dt.
\end{split}
\]

According to Gr\"onwall inequality in the integral form, we further have 
\[
\begin{split}
&\mathbb{E}[\overline{h}(\widetilde{\bm{\phi}}_{\bm{x}_0}^{\bm{w}}(T))]\leq \alpha(T)+\int_{0}^T \alpha(s) a e^{a(T-s)}ds\\
&= \overline{h}(\bm{x}_0)+\int_{0}^T \overline{h}(\bm{x}_0)a e^{a(T-s)} ds +bT+\int_{0}^{T} bsa e^{a(T-s)}ds\\
&~~~~~~~~~~+ (a-b)\int_{0}^T \mathbb{E}[1_{\mathcal{C}_1}(\widetilde{\bm{\phi}}_{\bm{x}_0}^{\bm{w}}(s))] ds \\
&~~~~~~~~~~+a(a-b)\int_{0}^T \int_{0}^s \mathbb{E}[1_{\mathcal{C}_1}(\widetilde{\bm{\phi}}_{\bm{x}_0}^{\bm{w}}(t))] dt e^{a(T-s)} ds    \\
&\leq \overline{h}(\bm{x}_0)e^{aT}-\frac{b}{a}+\frac{b}{a}e^{aT} \\
&~~~~~~~~~~~~~+(a-b)e^{aT}\mathbb{P}(\widetilde{\bm{\phi}}_{\bm{x}_0}^{\bm{w}}(T) \in \mathcal{C}_1)(-\frac{1}{a}e^{-aT}+\frac{1}{a}) 
\end{split}
\]
where $\alpha(s)=\overline{h}(\bm{x}_0)+\int_{0}^s (a-b)\mathbb{E}[1_{ \mathcal{C}_1}(\widetilde{\bm{\phi}}_{\bm{x}_0}^{\bm{w}}(t))]dt+\int_{0}^s b dt$. The last inequality is obtained according to  Lemma \ref{equiv}.

Thus, 
\[
\begin{split}
&-1\leq \mathbb{E}[\overline{h}(\widehat{\bm{\phi}}_{\bm{x}_0}^{\bm{w}}(T))]\leq \overline{h}(\bm{x}_0)e^{aT}-\frac{b}{a}+\frac{b}{a}e^{aT} \\
&~~~~~~~~~~~~~+(a-b)e^{aT}\mathbb{P}(\widetilde{\bm{\phi}}_{\bm{x}_0}^{\bm{w}}(T) \in  \mathcal{C}_1)(-\frac{1}{a}e^{-aT}+\frac{1}{a}) 
\end{split}
\]
After rearrangement, we have the conclusion that $\mathbb{P}(\widetilde{\bm{\phi}}_{\bm{x}_0}^{\bm{w}}(T)\in \mathcal{C}_1)\geq \max\{0,\frac{e^{aT}(h(\bm{x}_0)-\frac{b}{a})+\frac{b}{a}-1}{(1-\frac{b}{a})(e^{aT}-1)}\}$. Furthermore, according to Lemma \ref{equiv}, $\mathbb{P}_{\mathbb{T}}\geq  \max\{0,\frac{e^{aT}(h(\bm{x}_0)-\frac{b}{a})+\frac{b}{a}-1}{(1-\frac{b}{a})(e^{aT}-1)}\}$.

The conclusion that $\mathbb{P}_{\infty}\geq \max\{0,\frac{h(\bm{x}_0)-\frac{b}{a}}{1-\frac{b}{a}}\}$ can be obtained via letting $T$ approach infinity in $\frac{e^{aT}(h(\bm{x}_0)-\frac{b}{a})+\frac{b}{a}-1}{(1-\frac{b}{a})(e^{aT}-1)}$.
\end{proof}

The reason that $b$ is not allowed to be less than zero in Theorem \ref{exponential_s2} lies in (8), since the contradiction that $0\geq -b$ will be obtained over $\{\bm{x}\in \mathbb{R}^n \mid h(\bm{x})=0\}$ if $b\leq 0$.

\subsection{Safe Exit Controllers Synthesis Conditions for Problem II}
\label{sec:ecec1}
This subsection introduces the condition  to derive lower bounds on the exit probabilities in Problem II for both infinite and finite time horizons.

The condition introduced is an extension of the one \eqref{stochastic_c_e2} in Theorem \ref{exponential_s2}. Furthermore, in Problem I, leaving the set $\mathcal{C}$ for system \eqref{SDE} is guaranteed when it hits certain part of its boundary, i.e., $\mathcal{C}_1$. However, in the extended condition, hitting any state in the boundary of the set $\mathcal{C}$ implies that system \eqref{SDE} will leave the set $\mathcal{C}$. To accommodate this situation, the free parameter $b$ in the extended condition is allowed to be smaller than zero. This flexibility allows for a wider range of scenarios to be considered, expanding the feasibility of the condition and providing more general lower bounds on exit probabilities.

\begin{theorem}
\label{exponential_s}
Given a safe but uncomfortable set $\mathcal{C}$ defined  in Section \ref{sec:pre}, if there exists a locally Lipschitz controller $\bm{u}(\cdot): \overline{\mathcal{C}}\rightarrow \mathcal{U}$ satisfying the following condition:
\begin{equation}
\label{stochastic_c_e1}
\begin{cases}
&\mathcal{L}_{g,\bm{u}}(\bm{x})\geq  a g(\bm{x})-b,  \forall \bm{x}\in \mathcal{C},\\
&a>b,
\end{cases}
\end{equation}
then  for $\bm{x}_0\in \mathcal{C}$,
\begin{enumerate}
    \item when $a>0$, \[
    \begin{split}
    &\mathbb{P}(\exists t\in [0,T]. \bm{\phi}_{\bm{x}_0}^{\bm{w}}(t) \in \partial \mathcal{C})\\
    &\geq \max\{0,\frac{e^{aT}(g(\bm{x}_0)-\frac{b}{a})+\frac{b}{a}-1}{(1-\frac{b}{a})(e^{aT}-1)}\}
        \end{split}
        \] and  \[\mathbb{P}(\exists t\in [0,\infty). \bm{\phi}_{\bm{x}_0}^{\bm{w}}(t) \in \partial \mathcal{C})\geq \max\{0,\frac{g(\bm{x}_0)-\frac{b}{a}}{1-\frac{b}{a}}\}.\] 
\item when $a\leq 0$, \[\mathbb{P}(\exists t\in [0,T]. \bm{\phi}_{\bm{x}_0}^{\bm{w}}(t) \in \partial \mathcal{C})\geq \max\{0,1-\frac{g(\bm{x}_0)-1}{(b-a)T}\}\] and $\mathbb{P}(\exists t\in [0,\infty). \bm{\phi}_{\bm{x}_0}^{\bm{w}}(t) \in \partial \mathcal{C})=1.$
\end{enumerate}

\end{theorem}
\begin{proof}
 1).  The conclusion for $a>0$ can be obtained by following the proof in Theorem \ref{exponential_s2}, with $\mathcal{C}_1$ being replaced by $\partial \mathcal{C}$. 

2). Since $g(\bm{x})$ satisfies $g(\bm{x})\leq 1$ over $\overline{\mathcal{C}}$, we have $ag(\bm{x})-b\geq (a-b)>0$. Therefore, 
\[\widetilde{\mathcal{L}}_{g,\bm{u}}(\bm{x})+(a-b)1_{\partial \mathcal{C}}(\bm{x})\geq  a-b\geq 0,  \forall \bm{x}\in \overline{\mathcal{C}}.\]
Further, we conclude that 
\[\mathbb{E}[g(\widetilde{\bm{\phi}}_{\bm{x}_0}^{\bm{w}}(t))]\geq g(\bm{x}_0), \forall t\in [0,T]\] 
and 
\[
\begin{split}
\mathbb{E}[g(\widetilde{\bm{\phi}}_{\bm{x}_0}^{\bm{w}}(T))]&-g(\bm{x}_0)\\
&+(a-b)\int_{0}^{T}\mathbb{E}[1_{\partial \mathcal{C}}(\widetilde{\bm{\phi}}_{\bm{x}_0}^{\bm{w}}(t))] dt\geq (a-b)T.  
\end{split}
\]
According to Lemma \ref{equiv}, we further have  
\[(a-b)T\mathbb{P}(\exists t\in [0,T]. \widetilde{\bm{\phi}}_{\bm{x}_0}^{\bm{w}}(t)\in \partial \mathcal{C})\geq (a-b)T+g(\bm{x}_0)-1,\]
which implies $\mathbb{P}(\exists t\in [0,T]. \bm{\phi}_{\bm{x}_0}^{\bm{w}}(t) \in \partial \mathcal{C})\geq 1-\frac{g(\bm{x}_0)-1}{(b-a)T}$. Thus, we have 
$\mathbb{P}(\exists t\in [0,T]. \bm{\phi}_{\bm{x}_0}^{\bm{w}}(t) \in \partial \mathcal{C})\geq \max\{0,1-\frac{g(\bm{x}_0)-1}{(b-a)T}\}$. Via letting $T$ approach infinity, we further have $\mathbb{P}(\exists t\in [0,\infty). \bm{\phi}_{\bm{x}_0}^{\bm{w}}(t) \in \partial \mathcal{C})=1$.  
\end{proof}

\begin{figure*}[h!]
\centering
\subfloat[$w=10^{12}$]{
\begin{minipage}{0.23\linewidth}
\includegraphics[width=\textwidth,height=0.8\textwidth]{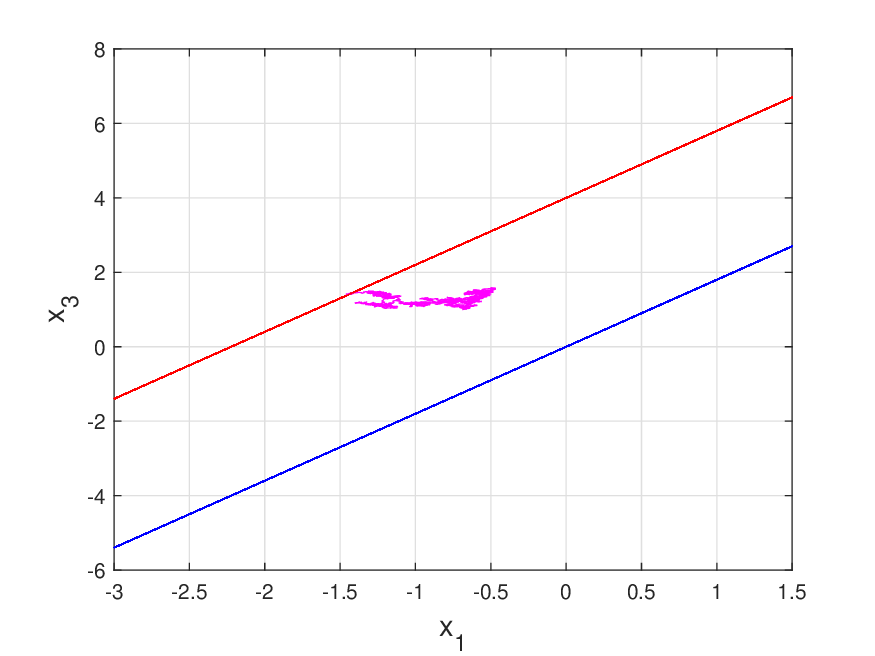}
\caption*{a-1}
\end{minipage}
\begin{minipage}{0.23\linewidth}
\includegraphics[width=\textwidth,height=0.8\textwidth]{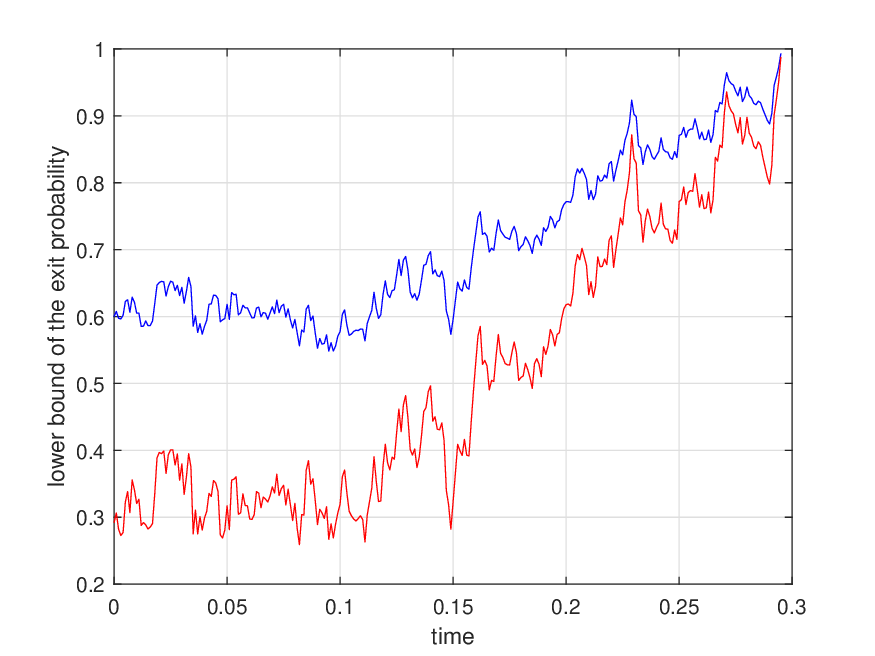}
\caption*{a-2}
\end{minipage}
\begin{minipage}{0.23\linewidth}
\includegraphics[width=\textwidth,height=0.8\textwidth]{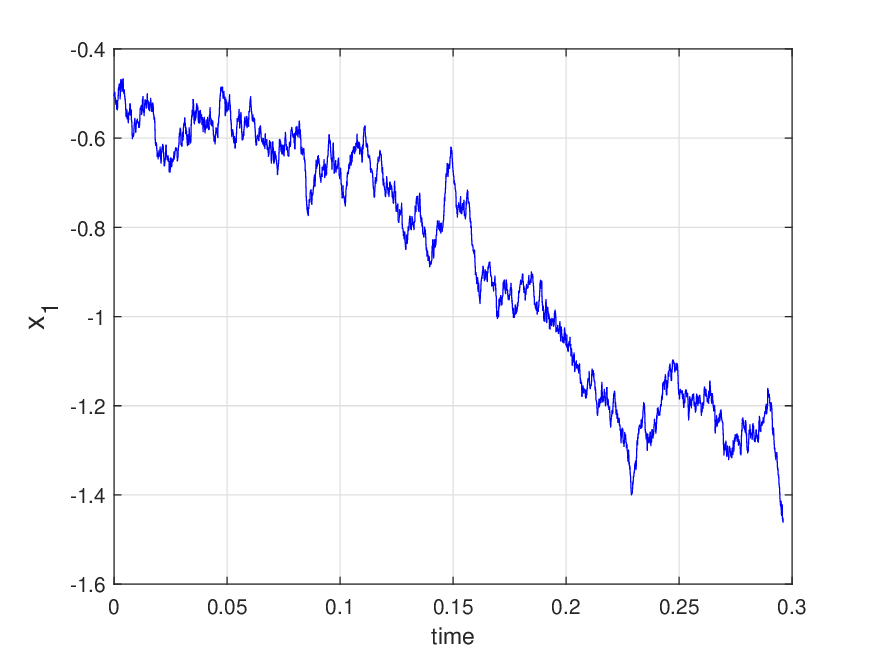}
\caption*{a-3}
\end{minipage}
\begin{minipage}{0.23\linewidth}
\includegraphics[width=\textwidth,height=0.8\textwidth]{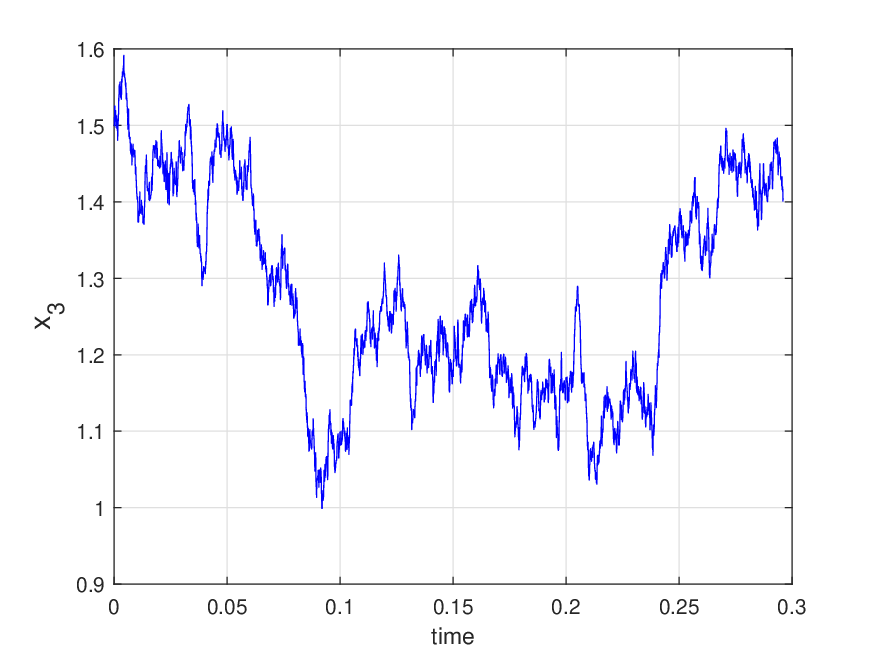}
\caption*{a-4}
\end{minipage}
}\\
\subfloat[$w=1$]{
\begin{minipage}{0.23\linewidth}
\includegraphics[width=\textwidth,height=0.8\textwidth]{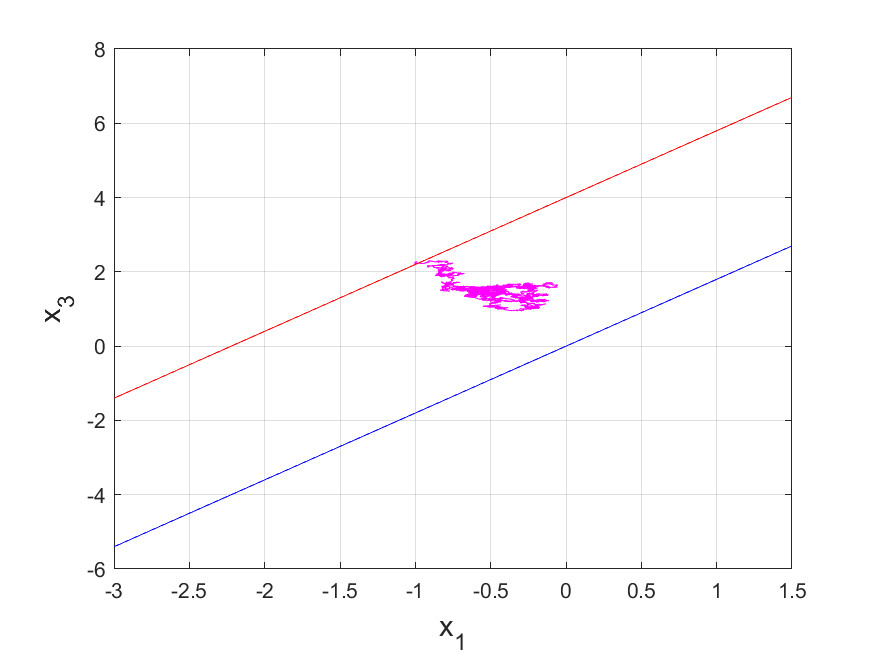}
\caption*{b-1}
\end{minipage}
\begin{minipage}{0.23\linewidth}
\includegraphics[width=\textwidth,height=0.8\textwidth]{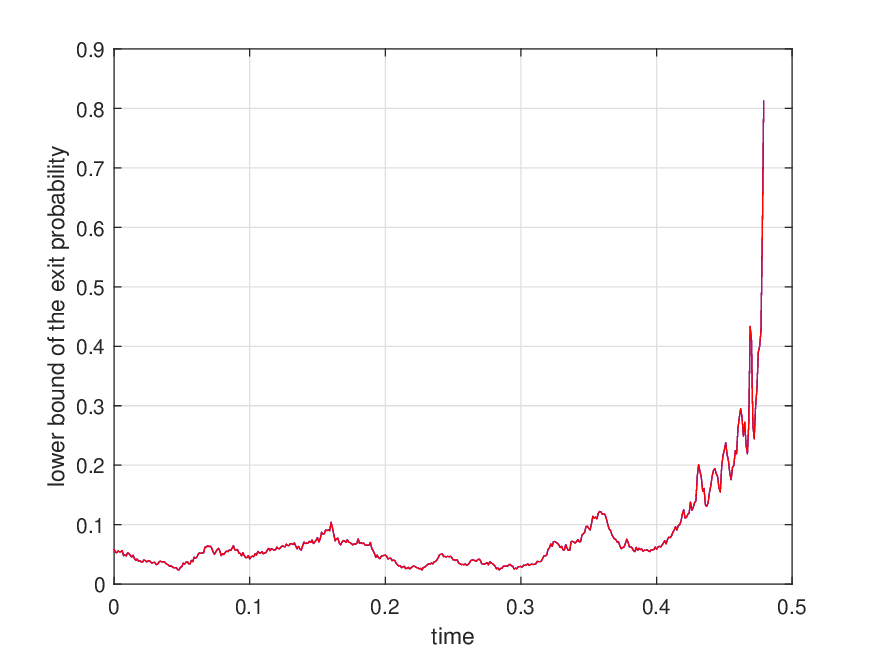}
\caption*{b-2}
\end{minipage}
\begin{minipage}{0.23\linewidth}
\includegraphics[width=\textwidth,height=0.8\textwidth]{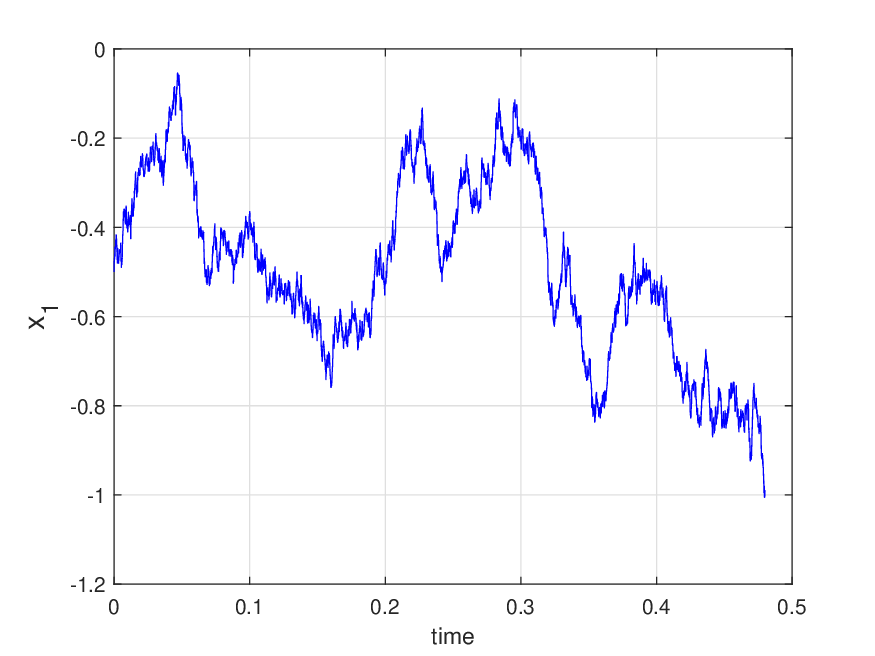}
\caption*{b-3}
\end{minipage}
\begin{minipage}{0.23\linewidth}
\includegraphics[width=\textwidth,height=0.8\textwidth]{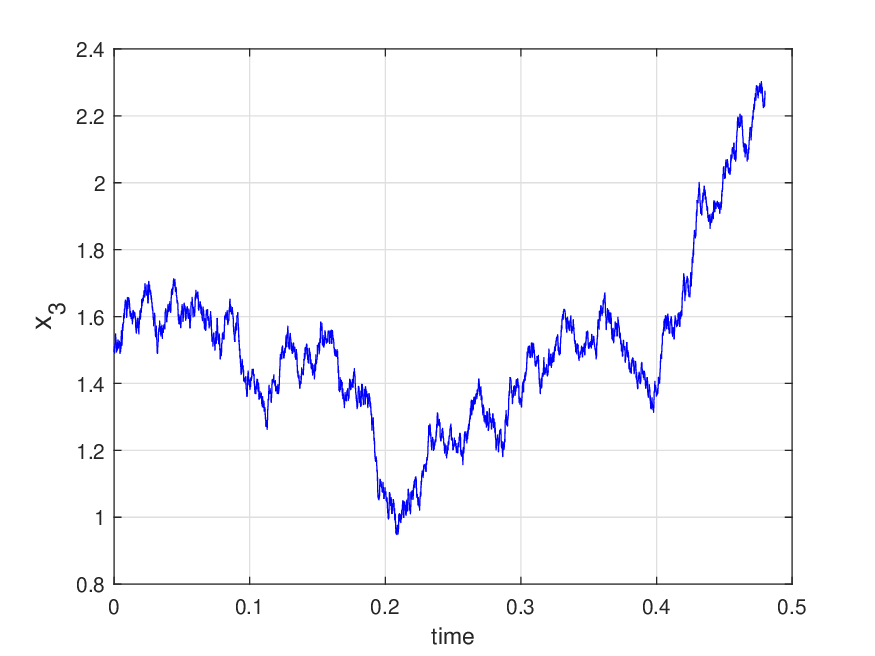}
\caption*{b-4}
\end{minipage}
}
 \caption{a-1 and b-1: red line -- $\mathcal{C}_1=\{(x_1,x_3)^{\top}\mid \frac{x_3-1.8x_1}{4}=1\}$, blue line -- $\{(x_1,x_3)^{\top}\mid \frac{x_3-1.8x_1}{4}=0\}$, magenta curve -- trajectory  driven by the controller computed via solving \eqref{s_qu}; a-2 and b-2: red curve -- lower bound of the exit probability when $\mathbb{T}=[0,2]$ with respect to time, i.e., $\frac{e^{a(T-t)}(h(\bm{x}(t))-\frac{b}{a})+\frac{b}{a}-1}{(1-\frac{b}{a})(e^{a(T-t)}-1)}$, blue curve -- lower bound of the exit probability when $\mathbb{T}=[0,\infty)$ with respect to time (blue and red curves collide in b-2), i.e., $\frac{h(\bm{x}(t))-\frac{b}{a}}{1-\frac{b}{a}}$; a-3 and b-3: $x_1(t)$; a-4 and b-4: $x_3(t)$.}
\label{fig_1}
\end{figure*}

\subsection{Linear-Program-Based Controllers}
\label{sec:lpbc}
In this subsection, we introduce our online linear programming based method for implicitly synthesizing optimal exit controllers that maximizes lower bounds of exit probabilities for both Problems I and II.

Except for the controller $\bm{u}$, both conditions \eqref{stochastic_c_e2} and \eqref{stochastic_c_e1} involve two additional free parameters, $a$ and $b$, that need to be determined in order to optimize the lower bound stated in Theorem \ref{exponential_s2} and \ref{exponential_s}. These conditions have a linear dependency on these parameters. However, the lower bounds exhibit nonlinearity  with respect to $a$ and $b$, except when $a\leq 0$ and $T\rightarrow \infty$ in Theorem \ref{exponential_s}. Hence, it is not advisable to solve a maximization problem with condition \eqref{stochastic_c_e2} (or \eqref{stochastic_c_e1}) and the lower bounds from Theorem \ref{exponential_s2} (or \ref{exponential_s}) as the objective function, especially for online optimization which which demands high efficiency.

On the other hand, it is observed that both the lower bounds $\frac{e^{aT}(h(\bm{x}_0)-\frac{b}{a})+\frac{b}{a}-1}{(1-\frac{b}{a})(e^{aT}-1)}$ and $\frac{h(\bm{x}_0)-\frac{b}{a}}{1-\frac{b}{a}}$ in Theorem \ref{exponential_s2} are monotonically increasing with $a$ and decreasing with $b$. Similar to Theorem \ref{exponential_s2}, all the lower bounds, i.e.,  $\frac{e^{aT}(g(\bm{x}_0)-\frac{b}{a})+\frac{b}{a}-1}{(1-\frac{b}{a})(e^{aT}-1)}$, $\frac{g(\bm{x}_0)-\frac{b}{a}}{1-\frac{b}{a}}$, $1-\frac{g(\bm{x}_0)-1}{(b-a)T}$, and $1$, in Theorem \ref{exponential_s2} are monotonically increasing with respect to $a$ and decreasing with respect to $b$. Moreover, it is observed that as the value of $a$ tends towards zero from the right, the lower bound  $\frac{e^{aT}(g(\bm{x}_0)-\frac{b}{a})+\frac{b}{a}-1}{(1-\frac{b}{a})(e^{aT}-1)}$ in Theorem \ref{exponential_s} tends to approach $1-\frac{g(\bm{x}_0)-1}{bT}$, i.e., $\lim_{a\rightarrow0^+}\frac{e^{aT}(g(\bm{x}_0)-\frac{b}{a})+\frac{b}{a}-1}{(1-\frac{b}{a})(e^{aT}-1)}=1-\frac{g(\bm{x}_0)-1}{bT}$, which is equal to the lower bound in the case of $a\leq 0$ with $a=0$.
Thus, the objective function $\max (a-wb)$ is a suitable candidate, where $w$ denotes a specified weighting factor. This factor allows for the adjustment of the relative importance of $b$ compared to $a$ according to the specific requirements of the problem. Additionally, in order to ensure boundedness of $a-wb$, we impose a bound constraint on $a$ and $b$. Consequently, an online linear program for implicitly synthesizing exit controllers of maximizing lower bounds of the exit probabilities for Problem I is formulated below.
\begin{equation} 
\label{s_qu}
\begin{split}
        &\max_{\bm{u}(\bm{x})\in \mathcal{U},a,b} a-w b\\
s.t.~&\mathcal{L}_{h,\bm{u}}(\bm{x})\geq a h(\bm{x})-b,\\
&a>b\geq 0, \\
&a\leq \delta,
\end{split}
\end{equation}
 where $\mathcal{L}_{h,\bm{u}}(\bm{x})=\frac{\partial h}{\partial \bm{x}}(\bm{f}_1(\bm{x})+\bm{f}_2(\bm{x})\bm{u}(\bm{x}))+\frac{1}{2}\textbf{tr}(\bm{\sigma}(\bm{x})^{\top}\frac{\partial^2 h}{\partial \bm{x}^2} \bm{\sigma}(\bm{x}))$ and $\delta>0$ is a specified bound.

 Correspondingly, an online linear program for implicitly synthesizing exit controllers of maximizing lower bounds of the exit probabilities for Problem II is formulated below.
 \begin{equation} 
\label{s_qu1}
\begin{split}
        &\max_{\bm{u}(\bm{x})\in \mathcal{U},a,b} a-w b\\
s.t.~&\mathcal{L}_{g,\bm{u}}(\bm{x})\geq a g(\bm{x})-b,\\
&a>b, \\
&a, b \in [-\delta,\delta],
\end{split}
\end{equation}
 where $\mathcal{L}_{g,\bm{u}}(\bm{x})=\frac{\partial g}{\partial \bm{x}}(\bm{f}_1(\bm{x})+\bm{f}_2(\bm{x})\bm{u}(\bm{x}))+\frac{1}{2}\textbf{tr}(\bm{\sigma}(\bm{x})^{\top}\frac{\partial^2 g}{\partial \bm{x}^2} \bm{\sigma}(\bm{x}))$ and $\delta>0$ is a specified  bound.

It is noteworthy that since a constraint on the control input, specifically $\bm{u}\in \mathcal{U}$, is  imposed, the existence of a solution for either of the optimization problems \eqref{s_qu} and \eqref{s_qu1} is not guaranteed. This is true even if the boundedness requirements on $a$ and $b$ are removed.

\begin{figure*}[h!]
\centering
\subfloat[$w=10^{12}$]{
\begin{minipage}{0.23\linewidth}
\includegraphics[width=\textwidth,height=0.8\textwidth]{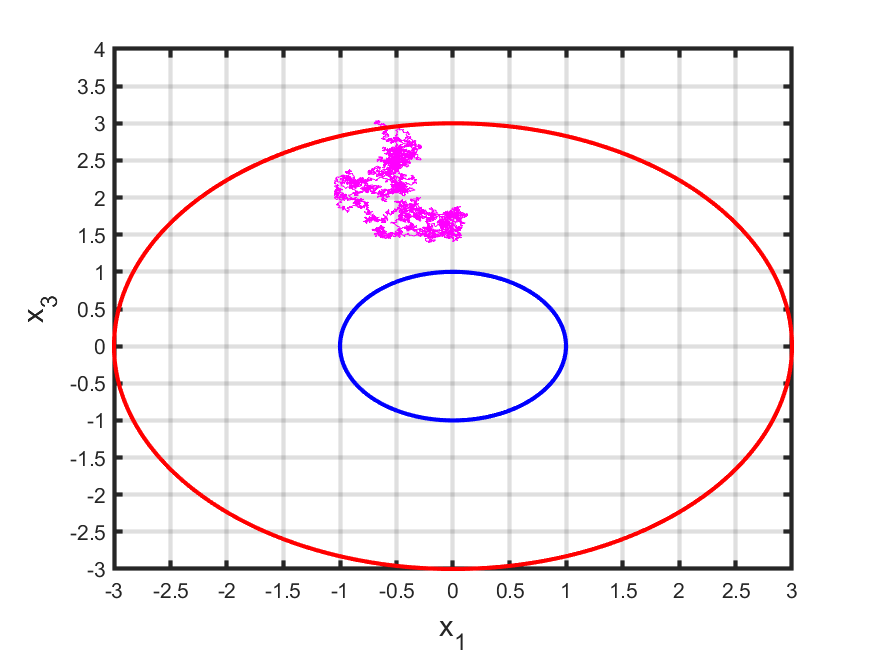}
\caption*{a-1}
\end{minipage}
\begin{minipage}{0.23\linewidth}
\includegraphics[width=\textwidth,height=0.8\textwidth]{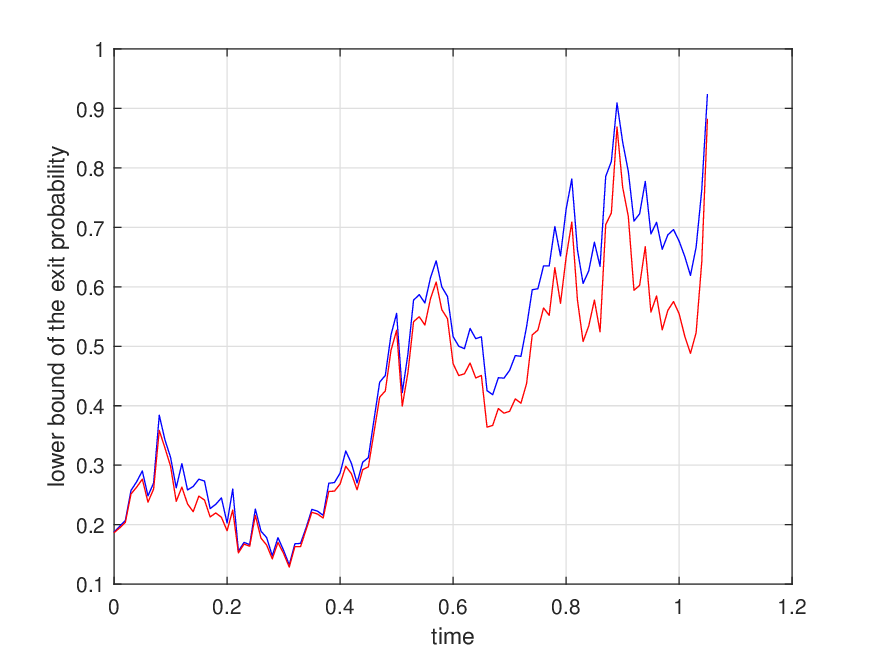}
\caption*{a-2}
\end{minipage}
\begin{minipage}{0.23\linewidth}
\includegraphics[width=\textwidth,height=0.8\textwidth]{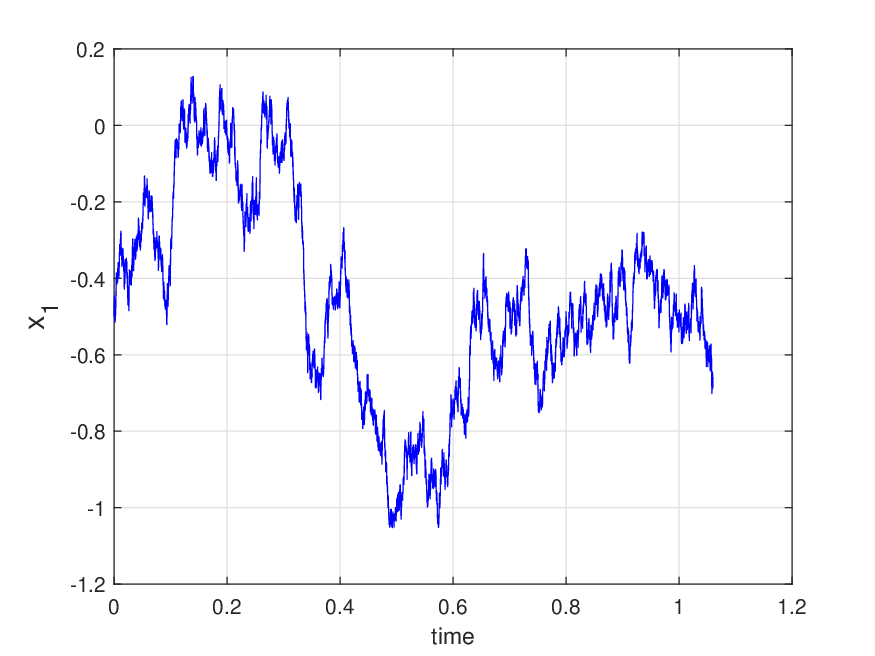}
\caption*{a-3}
\end{minipage}
\begin{minipage}{0.23\linewidth}
\includegraphics[width=\textwidth,height=0.8\textwidth]{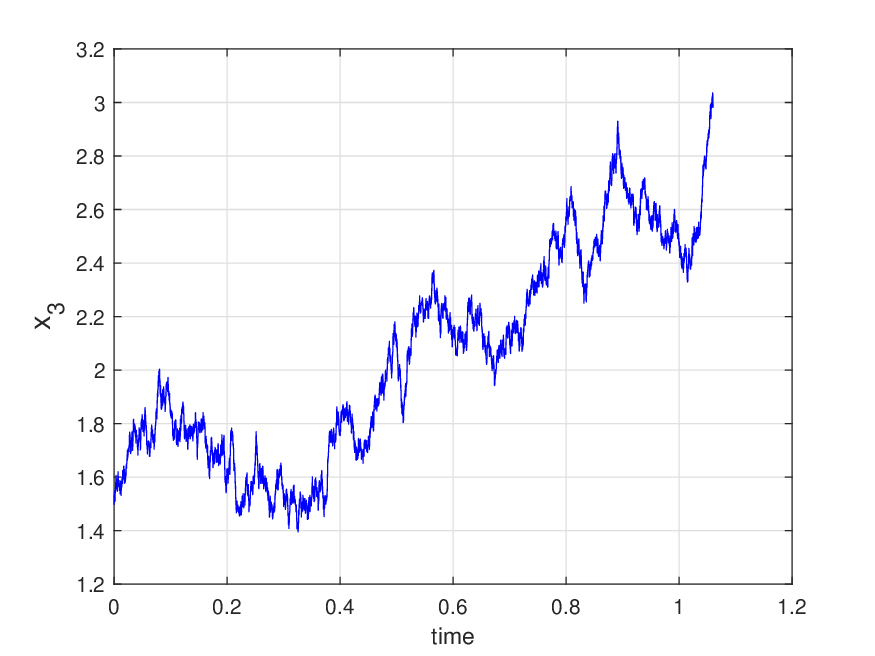}
\caption*{a-4}
\end{minipage}
}\\
\subfloat[$w=1$]{
\begin{minipage}{0.23\linewidth}
\includegraphics[width=\textwidth,height=0.8\textwidth]{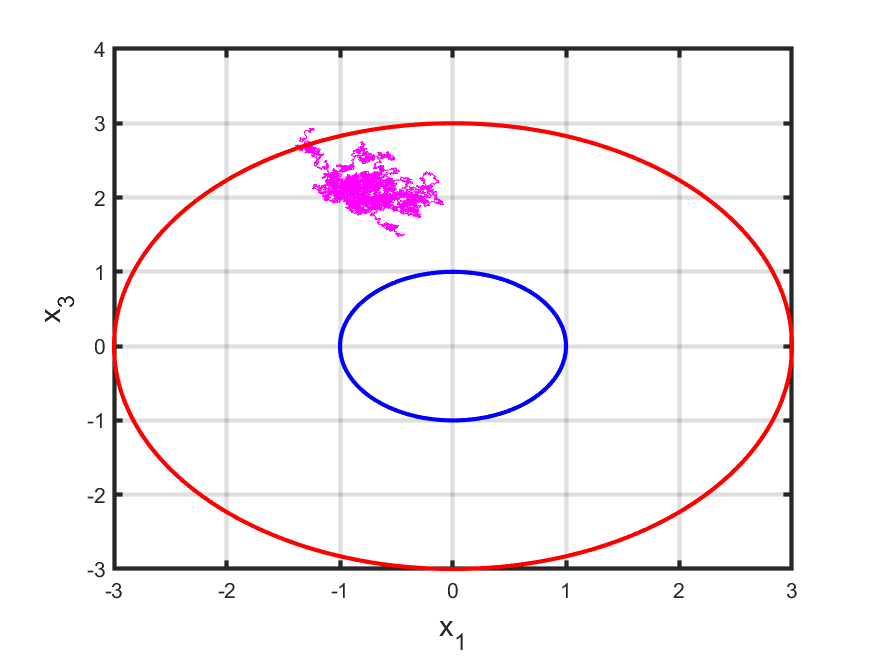}
\caption*{b-1}
\end{minipage}
\begin{minipage}{0.23\linewidth}
\includegraphics[width=\textwidth,height=0.8\textwidth]{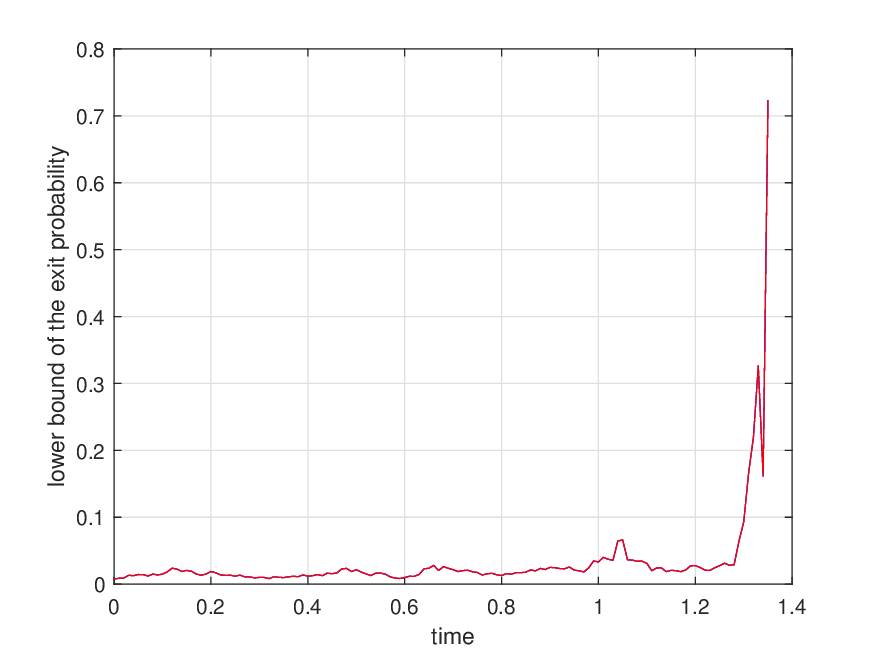}
\caption*{b-2}
\end{minipage}
\begin{minipage}{0.23\linewidth}
\includegraphics[width=\textwidth,height=0.8\textwidth]{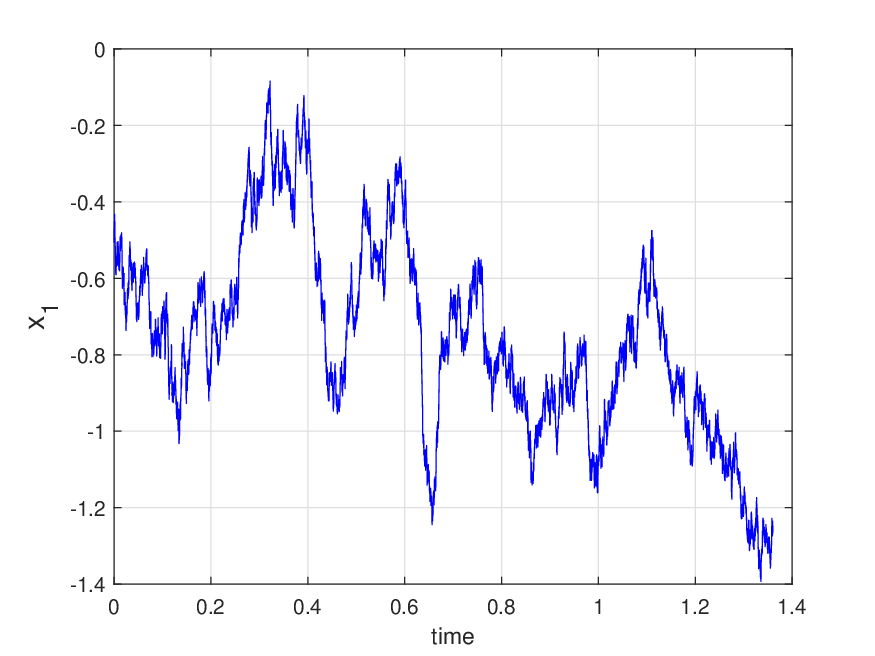}
\caption*{b-3}
\end{minipage}
\begin{minipage}{0.23\linewidth}
\includegraphics[width=\textwidth,height=0.8\textwidth]{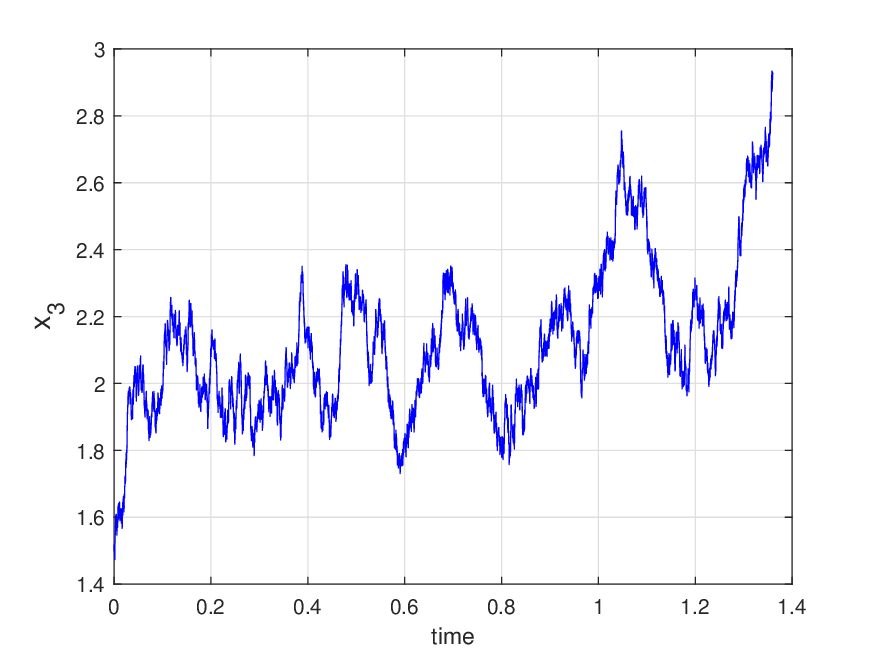}
\caption*{b-4}
\end{minipage}
    }
\caption{a-1 and b-1: red line -- $\mathcal{C}_1=\{(x_1,x_3)^{\top}\mid \frac{x_1^2+x_3^2-1}{8}=1\}$, blue line -- $\{(x_1,x_3)^{\top}\mid \frac{x_1^2+x_3^2-1}{8}=0\}$, magenta curve -- trajectory  driven by the controller computed via solving \eqref{s_qu}; a-2 and b-2: red curve -- lower bound of the exit probability when $\mathbb{T}=[0,2]$ with respect to time, i.e., $\frac{e^{a(T-t)}(h(\bm{x}(t))-\frac{b}{a})+\frac{b}{a}-1}{(1-\frac{b}{a})(e^{a(T-t)}-1)}$, blue curve -- lower bound of the exit probability when $\mathbb{T}=[0,\infty)$ with respect to time (blue and red curves collide in b-2), i.e., $\frac{h(\bm{x}(t))-\frac{b}{a}}{1-\frac{b}{a}}$; a-3 and b-3: $x_1(t)$; a-4 and b-4: $x_3(t)$.}
\label{fig_3}
\end{figure*}

\section{Examples}
\label{sec:expe}
In this section we demonstrate our linear-programming based exit controllers synthesis method on one example involving three scenarios.

    Consider a system with three states $(x_1,x_2,x_3)^{\top}$\cite{wang2021safety,clark2019control}, where $x_1$ denotes the velocity of the following vehicle, $x_2$ denotes the velocity of the leading vehicle, and $x_3$ denotes the distance between the vehicles. The velocity of the leading vehicle was chosen as a constant. The input is the force applied to the following vehicle, leading to dynamics
\[
d\begin{pmatrix}
    x_1\\
    x_2\\
    x_3
\end{pmatrix}
=\begin{pmatrix}
    -F_r(\bm{x})/M\\
    0\\
    x_2-x_1
\end{pmatrix}
+\begin{pmatrix}
    1/M\\
    0\\
    0
\end{pmatrix}\bm{u}
+\sum dW,
\]
where $F_r=f_0+f_1x_1+f_2 x_1^2$ is the aerodynamic drag with constants $f_0=0.1,f_1=5$, and $f_1=0.25$. The mass $M=1650$, $\sum=\begin{pmatrix}
    1&0&0\\
    0&0&0\\
    0&0&1
\end{pmatrix}$, and $u\in [-1,1]$. The initial state for $x_2$ was chosen as $x_2(0)=0.5$. Since the velocity of the leading vehicle was chosen as a constant, the system is equivalently reduced to
\[
d\begin{pmatrix}
    x_1\\
    x_3
\end{pmatrix}
=\begin{pmatrix}
    -F_r(\bm{x})/M\\
    0.5-x_1
\end{pmatrix}
+\begin{pmatrix}
    1/M\\
    0
\end{pmatrix}\bm{u}
+\begin{pmatrix}
    1& 0\\
    0&1
\end{pmatrix}dW.
\]

We consider three  scenarios with both the finite time horizon of $\mathbb{T}=[0,2]$ and the infinite time horizon of $\mathbb{T}=[0,\infty)$. Moreover, the weighting factor $w$ in the optimization problems \eqref{s_qu} and \eqref{s_qu1} is chosen to be $10^{12}$ and $1$. The first two scenarios correspond to Problem I. The first scenario features an unbounded safe set $\mathcal{S}$ and uncomfortable set $\mathcal{C}$, while the second one features an unbounded safe set $\mathcal{S}$ but a bounded uncomfortable set $\mathcal{C}$. The third scenario corresponds to Problem II, which includes a bounded uncomfortable set $\mathcal{C}$. Detailed configuration information and some computation results are presented below.
\begin{enumerate}
\item The safe set is $\mathcal{S}=\{(x_1,x_3)^{\top}\mid x_3-1.8x_1>0\}$, the safe but uncomfortable set is defined as $\mathcal{C}=\{(x_1,x_3)^{\top}\mid 0<\frac{x_3-1.8x_1}{4}<1\}$, and the initial state is set to $(-0.5,1.5)^{\top}$. The simulation trajectories and lower bounds of exit probabilities, computed by solving the linear optimization \eqref{s_qu} with $\delta=10$, are presented in Fig. \ref{fig_1}. 

\item The safe set is $\mathcal{S}=\{(x_1,x_3)^{\top}\mid x_1^2+x_3^2-1>0\}$, the safe but uncomfortable set is $\mathcal{C}=\{(x_1,x_3)^{\top}\mid 0<\frac{x_1^2+x_3^2-1}{8}<1\}$, and the initial state is $(-0.5,1.5)^{\top}$.  The simulation trajectories and lower bounds of exit probabilities, computed by solving the linear optimization \eqref{s_qu} with $\delta=10$, are presented in Fig. \ref{fig_3}.

\item The safe set is $\mathcal{S}=\{(x_1,x_3)^{\top}\mid x_1^2+x_3^2>1\}$, the uncomfortable set is $\mathcal{C}=\{(x_1,x_3)^{\top}\mid \frac{(x_1-10)^2+(x_3-10)^2}{64}<1\}$, and the initial state is $(10,10)^{\top}$. The simulation trajectories and lower bounds of exit probabilities, computed by solving optimization \eqref{s_qu1} with $\delta=10$, are presented in Fig. \ref{fig_5}. 

\end{enumerate}

The results presented in Figures \ref{fig_1}, \ref{fig_3}, and \ref{fig_5} demonstrate the significant impact of the weighting factor $w$ on the performance of the synthesized controllers. Notably, in the first two scenarios, the controllers computed with $w=10^{12}$ show superior performance in safely guiding the system out of the uncomfortable set $\mathcal{C}$ with high probabilities, compared to those obtained with $w=1$, especially during the initial phase. In the third scenario, where $\partial \mathcal{C}\cap \partial \mathcal{S}=\emptyset$, the controller synthesized with $w=10^{12}$ exhibits superior performance in terms of achieving high probabilities for safely driving the system out of the uncomfortable set $\mathcal{C}$, when an infinite time horizon $\mathbb{T}=[0,\infty)$ is considered. However, the performance during the initial phase is inferior when the time horizon is limited to $\mathbb{T}=[0,2]$, compared to that obtained with $w=1$.

\begin{figure*}[h!]
\centering
\subfloat[$w=10^{12}$]{
\begin{minipage}{0.23\linewidth}
\includegraphics[width=\textwidth,height=0.7\textwidth]{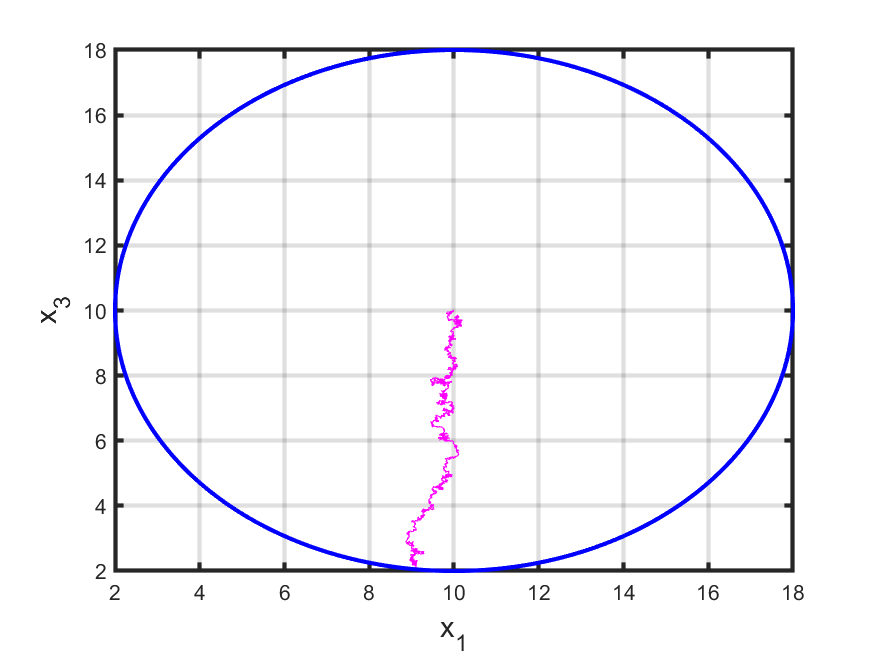}
\caption*{a-1}
\end{minipage}
\begin{minipage}{0.23\linewidth}
\includegraphics[width=\textwidth,height=0.7\textwidth]{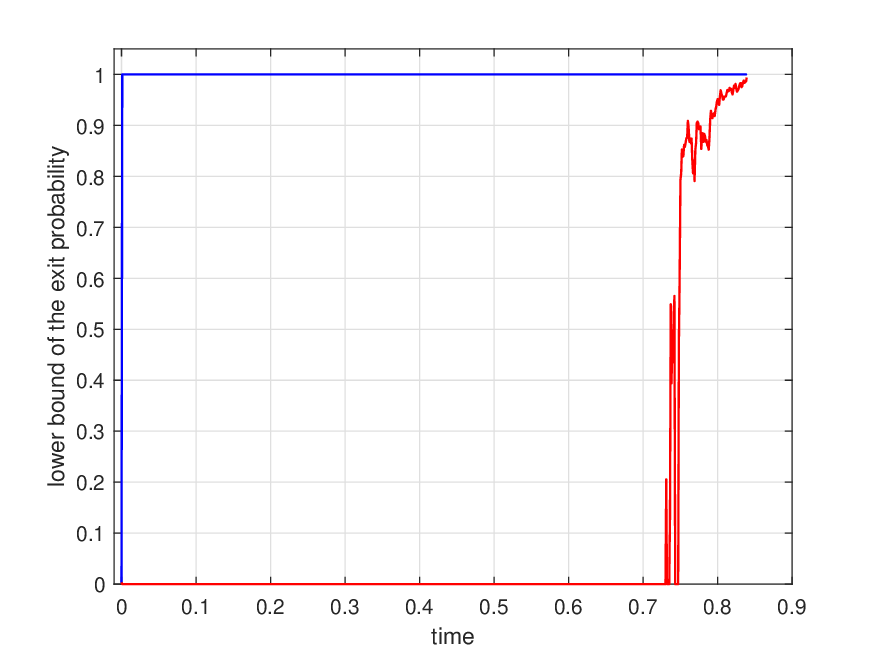}
\caption*{a-2}
\end{minipage}
\begin{minipage}{0.23\linewidth}
\includegraphics[width=\textwidth,height=0.7\textwidth]{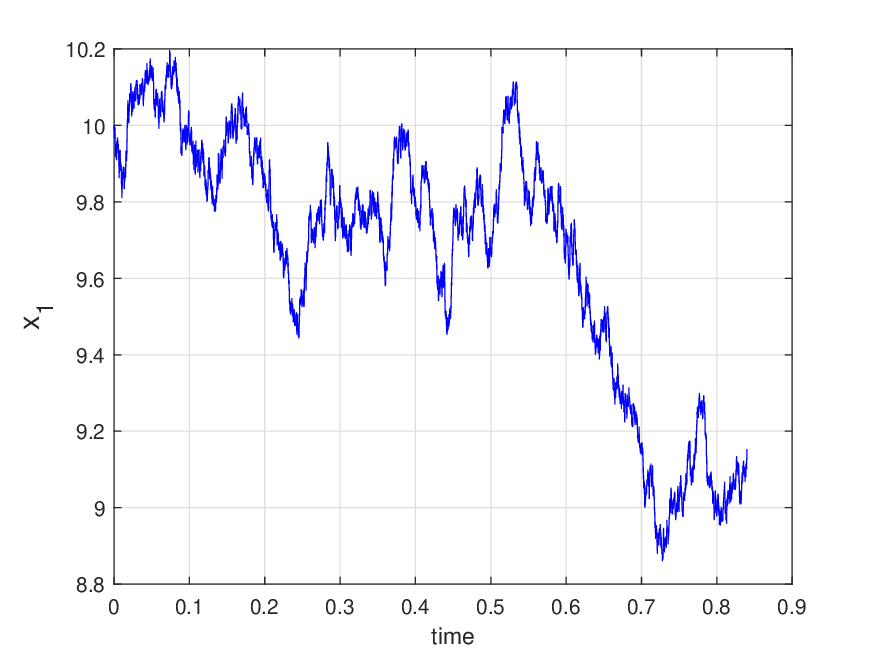}
\caption*{a-3}
\end{minipage}
\begin{minipage}{0.23\linewidth}
\includegraphics[width=\textwidth,height=0.7\textwidth]{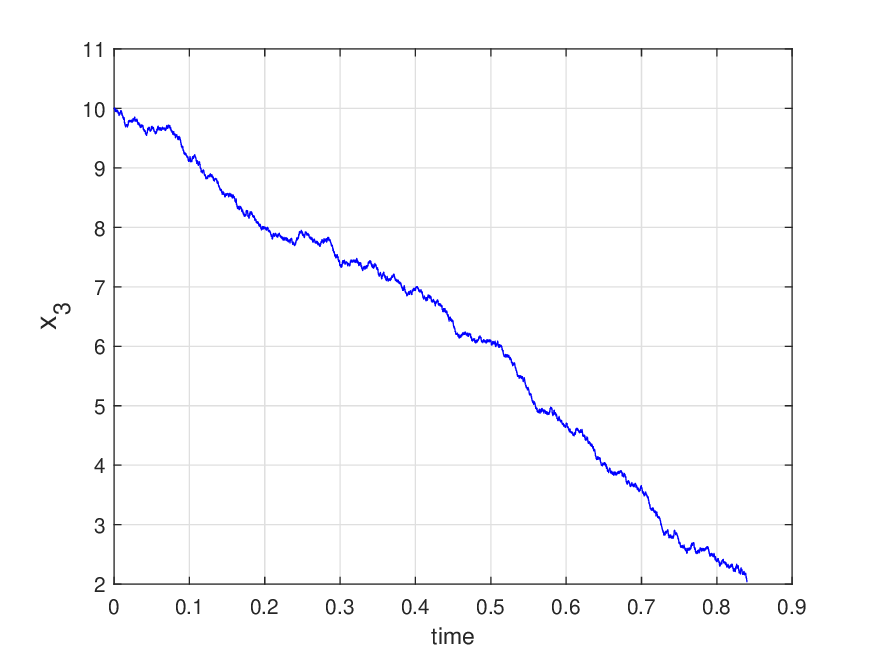}
\caption*{a-4}
\end{minipage}
}\\
\subfloat[$w=1$]{
\begin{minipage}{0.23\linewidth}
\includegraphics[width=0.8\textwidth,height=0.7\textwidth]{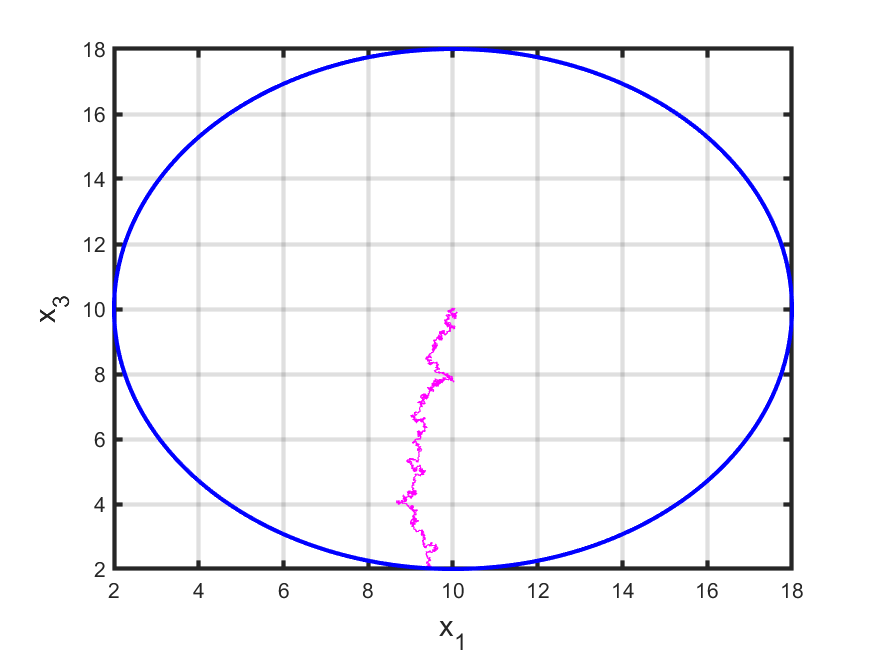}
\caption*{b-1}
\end{minipage}
\begin{minipage}{0.23\linewidth}
\includegraphics[width=\textwidth,height=0.7\textwidth]{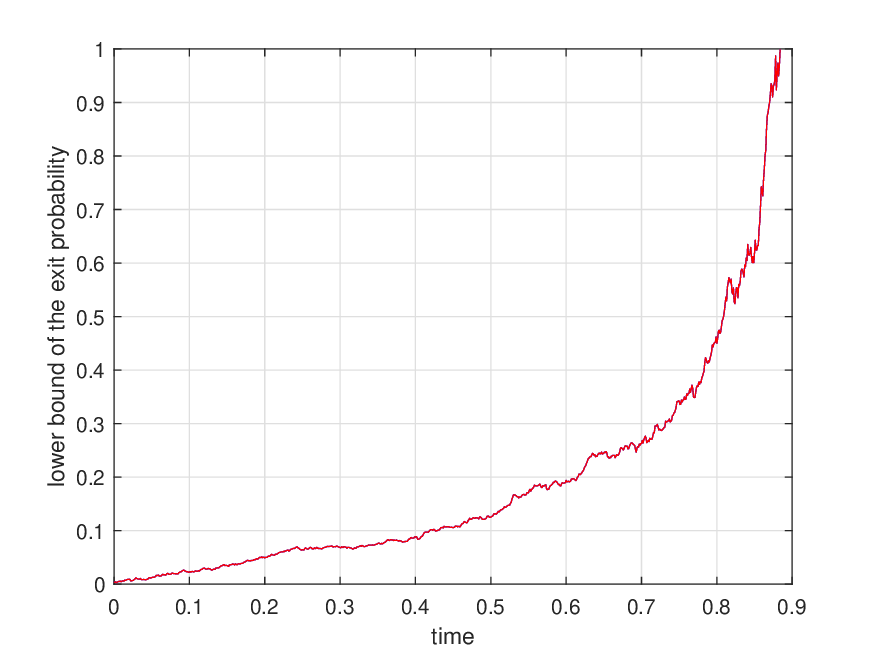}
\caption*{b-2}
\end{minipage}
\begin{minipage}{0.23\linewidth}
\includegraphics[width=\textwidth,height=0.7\textwidth]{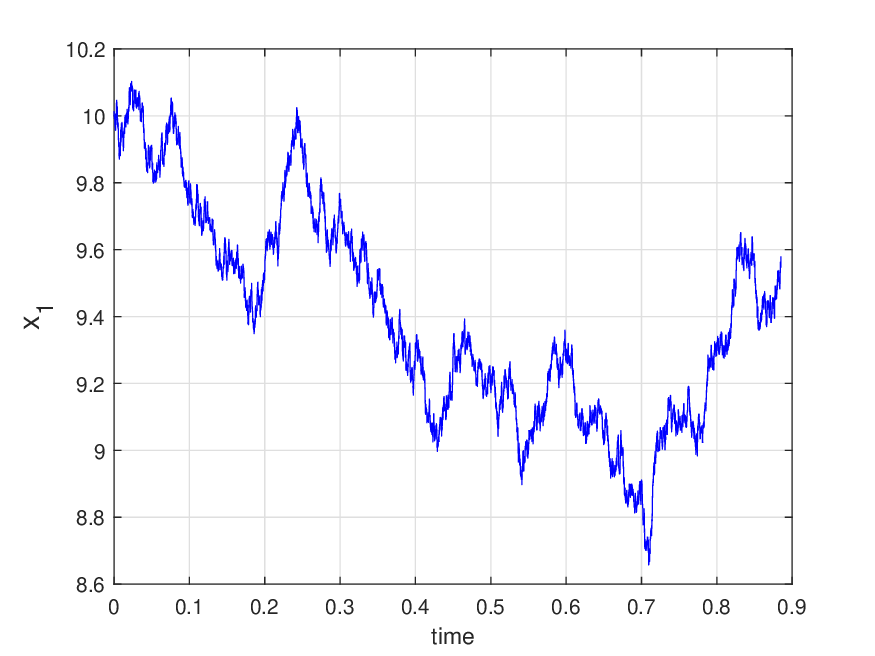}
\caption*{b-3}
\end{minipage}
\begin{minipage}{0.23\linewidth}
\includegraphics[width=\textwidth,height=0.7\textwidth]{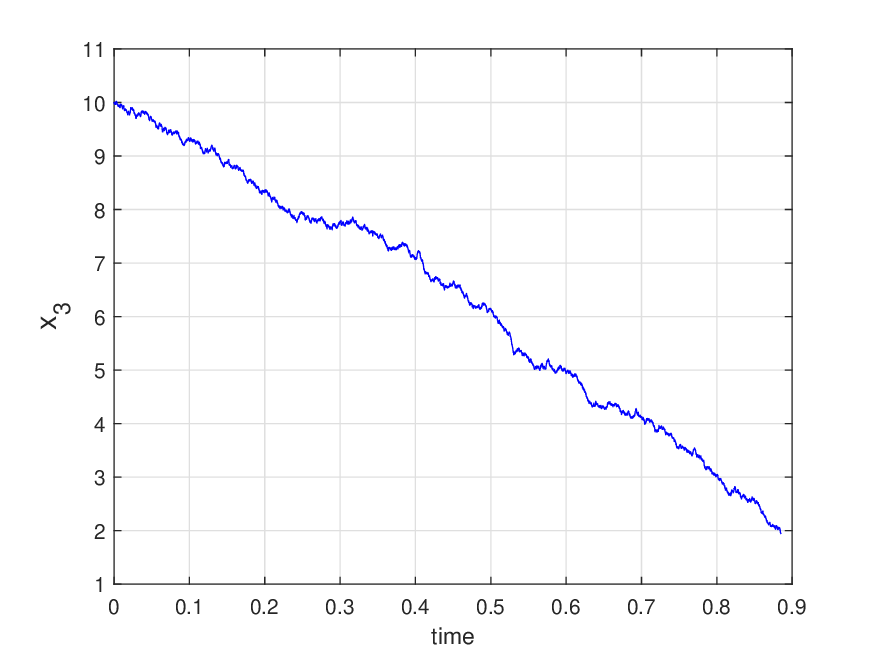}
\caption*{b-4}
\end{minipage}
}
\caption{a-1 and b-1: blue curve -- $\{(x_1,x_3)^{\top}\mid \frac{(x_1-1))^2+(x_3-10)^2}{64}=1\}$, magenta curve -- trajectory  driven by the controller computed via solving \eqref{s_qu1}; a-2 and b-2: red curve -- lower bound of the exit probability when $\mathbb{T}=[0,2]$ with respect to time, i.e., $\frac{e^{a(T-t)}(g(\bm{x}(t))-\frac{b}{a})+\frac{b}{a}-1}{(1-\frac{b}{a})(e^{a(T-t)}-1)}$ if $a>0$ and $1-\frac{g(\bm{x}(t))-1}{(b-a)(T-t)}$ if $a\leq 0$, blue curve -- lower bound of the exit probability when $\mathbb{T}=[0,\infty)$ with respect to time (blue and red curves collide in b-2), i.e., $\frac{h(\bm{x}(t))-\frac{b}{a}}{1-\frac{b}{a}}$ if $a>0$ and $1$ if $a\leq 0$; a-3 and b-3: $x_1(t)$; a-4 and b-4: $x_3(t)$.}
\label{fig_5}
\end{figure*}

\section{Conclusion}
\label{sec:con}
This paper focused on the synthesis of safe exit controllers for continuous-time systems described by SDEs. The main objective is to design controllers that maximize the lower bounds of the exit probability that the system escapes from a safe but uncomfortable set within a specific time horizon and enters a comfortable set. The paper discussed two cases: the first case involves the scenario where the boundary of the safe set is a subset of the boundary of the safe but uncomfortable set, and the second case deals with situations where the boundaries do not intersect. In the first case, the paper presented a sufficient condition for lower-bounding the exit probability. This condition provides a guideline for constructing online linear programming problems, which in turn facilitate synthesizing optimal exit controllers implicitly. These controllers are designed to maximize the lower bounds of the exit probabilities. Then, this sufficient condition was extended  to the second case, where the boundaries of the safe set and the uncomfortable set do not intersect. Finally, an example was presented to validate the proposed method. 

The first case discussed in this paper involves a scenario where the boundary of the safe set intersects with that of the uncomfortable set. However, it is limited to the typical case where the boundary of the safe set is a subset of the boundary of the uncomfortable set. In future studies, we will explore more general cases where only a subset of the safe set's boundary intersects that of the uncomfortable set.

\bibliographystyle{abbrv}
\bibliography{ref}
\end{document}